\documentclass[11pt]{article}
\usepackage[left=1.02in,right=1.04in,top=1in,bottom=1.04in,includefoot,footskip=18pt]{geometry}
\usepackage{amsmath,amssymb,amsthm,mathtools}
\usepackage{mathrsfs}
\usepackage{booktabs,array,longtable}
\usepackage[T1]{fontenc}
\usepackage{lmodern}
\usepackage{authblk}
\usepackage{quantikz}
\usepackage{xcolor}
\usepackage{subcaption}
\usepackage{hyperref}
\hypersetup{hidelinks,pdflang={en-US},
 pdftitle={Efficient Non-Uniform Quantum Hermite Transform through Adaptive Sampling},
 pdfauthor={Nitay Mayo; Aryeh Lev Zabokritskiy (Yohananov)}}
\newtheorem{theorem}{Theorem}
\newtheorem{lemma}[theorem]{Lemma}
\newtheorem{proposition}[theorem]{Proposition}

\theoremstyle{remark}
\newcommand{\id}{\mathbb I}
\newcommand{\hilbert}{\mathcal H}
\providecommand{\ket}[1]{}
\providecommand{\bra}[1]{}
\providecommand{\braket}[2]{}
\renewcommand{\ket}[1]{\lvert #1\rangle}
\renewcommand{\bra}[1]{\langle #1\rvert}
\renewcommand{\braket}[2]{\langle #1\mid #2\rangle}

\DeclareMathOperator{\diag}{diag}

\title{Efficient Non-Uniform Quantum Hermite Transform through Adaptive Sampling}
\author[1]{Nitay Mayo}
\author[2]{Aryeh Lev Zabokritskiy (Yohananov)\thanks{ORCID: \href{https://orcid.org/0000-0003-3151-6192}{0000-0003-3151-6192}}}
\affil[1]{Department of Computer Science, Technion -- Israel Institute of Technology, Haifa, Israel}
\affil[2]{Department of Computer Science, Tel-Hai University of Kiryat Shmona in the Galilee; MIGAL -- Galilee Research Institute}
\date{}
\begin{document}

\maketitle
\begin{abstract}
On the span of the first $N$ oscillator modes, Gauss--Hermite quadrature gives an exact change of basis between mode coefficients and $N$ weighted position space samples. We implement this transform with $O(N\operatorname{polylog}(N,1/\varepsilon))$ logical gates and polylogarithmic quantum width. The operator-error bound $\varepsilon$ holds on arbitrary superpositions and includes all auxiliary registers. The construction uses signed averages on adaptive windows to convert uniform-grid samples into weighted Hermite-root samples. Their varying widths control the amplification cost, giving the near-linear bound.
\end{abstract}

\section{Introduction}
Gaussian quadrature allows certain continuous integrals to be evaluated exactly from finitely many samples. For arbitrary combinations of the first $N$ oscillator modes, $N$ Hermite roots and their quadrature weights preserve inner products and recover every continuous mode coefficient exactly. Thus the finite sample representation introduces no discretization error on this space. The question here is the quantum cost of this change of basis. It also provides a finite representation of harmonic-oscillator evolution and fractional Fourier rotations, both diagonal in the oscillator basis.

Golub and Welsch give the classical eigenvalue construction of Gaussian quadrature~\cite{golub1969}. Pli\'s and Zak implement the finite Hermite transform on a quantum computer \cite{plis2025}; Webb and Maierhofer study stable classical construction of the normalized Hermite matrix \cite{webb2026}. Our contribution is a near-linear gate bound for the same finite transform. We convert an existing efficient Hermite sampler on a uniform grid~\cite{jain2025} by using signed averages on disjoint adaptive windows. Adapting their widths to the local root spacing controls normalization even at the outermost roots.

Our conventions are as follows. Let $N=2^n$, where $n\geq1$ is the number of data qubits. We denote the finite Hermite transform from weighted root samples to oscillator coefficients by $B$ and construct its inverse $B^\dagger$. Circuit errors are operator-norm errors that include every auxiliary register, with fixed phase. The input is an arbitrary quantum superposition; loading classical data into its amplitudes and reading out all coefficients are separate tasks.

\subsection{The finite transform and its input}
We define the oscillator number basis in $L^2(\mathbb R)$ using the normalized Hermite functions:
\begin{equation}
 \varphi_k(x)=e^{-x^2/2} \frac{H_k(x)}{\sqrt{2^k k!\sqrt\pi}},\qquad k\geq0,
\end{equation}
where $H_k$ are the physicists' Hermite polynomials (with $H_0(x)=1$). The corresponding finite subspace is $\hilbert_N=\operatorname{span}\{\varphi_0,\ldots,\varphi_{N-1}\}$.

While continuous states can be formally mapped to a phase-space representation via the standard Bargmann--Fock continuous kernel \cite{Bargmann1961}, $B$ outputs number-basis coefficients rather than values on a complex grid; the circuit and its reverse implement the two directions of this finite change of basis.

Let $x_1<\cdots<x_N$ be the roots of $H_N$, with corresponding Gauss--Hermite quadrature weights~\cite{dlmf}:
\begin{equation}
 w_i=\frac{2^{N-1}N!\sqrt\pi}{N^2H_{N-1}(x_i)^2}>0.
\end{equation}
For a polynomial $p$ of degree at most $2N-1$, this quadrature exactly evaluates the Gaussian integral:
\begin{equation}\label{eq:quadrature_approx}
 \int_{\mathbb R}e^{-x^2}p(x)\,dx=\sum_{i=1}^Nw_i p(x_i).
\end{equation}

Define the elements of the transform matrix $B$ as~\cite{plis2025, webb2026}:
\begin{equation}\label{eq:B_eq}
 B_{k,i}=\frac{\sqrt{w_i}}{\sqrt{2^k k!\sqrt\pi}}H_k(x_i),
 \qquad 0\leq k<N,\quad 1\leq i\leq N.
\end{equation}
Rows are number-basis indices; column $i$ is stored under computational label $i-1$. Applying the exact quadrature to $H_kH_m$, of degree at most $2N-2$, yields orthonormality~\cite[Sec.~3.5(v)]{dlmf}:
{\begin{equation}\label{eq:unitary_B}
 BB^\dagger=B^\dagger B=\id.
\end{equation}}

The sampling convention is essential. For a continuous function with defined point values, we encode the amplitudes as weighted samples:
\begin{equation}\label{eq:weighted-input}
 (S_N\psi)_i=\sqrt{w_i}\,e^{x_i^2/2}\psi(x_i).
\end{equation}
For any state $\psi=\sum_{k=0}^{N-1}c_k\varphi_k\in\hilbert_N$, these weights ensure the exact recovery of both coefficients and norms:
\begin{equation}\label{eq:exact-quadrature-map}
 S_N\psi=B^\dagger \vec c,\qquad BS_N\psi=\vec c,\qquad
 \|S_N\psi\|_2=\|\psi\|_{L^2}.
\end{equation}
These identities express the core quadrature advantage: the encoded amplitudes are $S_N\psi$, not the unweighted values $\psi(x_i)$. The corresponding states $\ket{x_i}$ form an orthonormal finite position basis, rather than a set of continuous position eigenstates.
For inputs outside $\hilbert_N$, truncation and aliasing from the sampled tail contribute additional errors.

\section{Main result}
For the main result take $N=2^n\geq2$ and $0<\varepsilon\leq1/2$, and let the input and output data registers have $n$ qubits. Let $E_{\mathrm{in}}$ and $E_{\mathrm{out}}$ denote their isometric embeddings into the full circuit space with every other register initialized to zero. Fixed register permutations identify their data labels with the mode and root labels, respectively.

\begin{theorem}[Near-linear implementation]\label{thm:near-linear}
There is a Clifford+$T$ quantum circuit $U_{N,\varepsilon}$
such that
\begin{equation}\label{eq:aw-main-contract}
 \bigl\|U_{N,\varepsilon}E_{\mathrm{in}}
       -E_{\mathrm{out}}B^\dagger\bigr\|\leq\varepsilon.
\end{equation}
The circuit uses $O\bigl(N\operatorname{polylog}(N,1/\varepsilon)\bigr)$ logical gates, hence at most this depth, and $\operatorname{polylog}(N,1/\varepsilon)$ qubits including workspace. Its reverse implements $B$ with the same error and resources.
\end{theorem}

Here $\operatorname{polylog}(N,1/\varepsilon)$ denotes a fixed polynomial
in $1+\log N+\log(1/\varepsilon)$. The bound is for quantum execution after classical compilation; it excludes arbitrary input loading and full classical readout. For $N=1$ the transform is the identity.
Supplementary Section~\ref{S-sec:input-error} gives the separate continuous-input error bounds.

While Theorem~\ref{thm:near-linear} establishes the headline complexity bound for the transform, the main construction is the efficient averaging window operator $R$. Generic synthesis gives a quadratic gate upper bound, but does not exploit the spacing and weights of the Hermite roots. The operator $R$ maps uniform-grid samples to the non-uniform weighted root samples. The near-linear block-encoding of this window operator enables the efficient implementation of the transform $B$.

The uniform-grid transform used in the proof is based on~\cite{jain2025}; its precise sampling and phase guarantees are proved in Supplementary Theorem~\ref{S-thm:uniform-grid}.

\section{Applications and Use-Cases}
\label{sec:applications}

The transform $B$ connects weighted position samples to the Hermite basis, in which both harmonic evolution and fractional Fourier transformation are diagonal. We describe these two applications. The exact identities below concern normalized states in the finite space $\hilbert_N$; implementing them requires the finite-precision circuits of Theorem~\ref{thm:near-linear}.

\subsection{Fast-Forwarding the Quantum Harmonic Oscillator}
In oscillator units, the quantum harmonic oscillator has Hamiltonian $H_{\rm osc}=(-d^2/dx^2+x^2)/2$ and eigenvalues $E_k=k+1/2$. Its restriction to $\hilbert_N$ can be evolved directly for any time $t$, without a time-stepping or product-formula approximation. By~\eqref{eq:exact-quadrature-map}, the corresponding operator on weighted samples is
{\begin{equation}\label{eq:qho-finite-evolution}
 U_X(t)=B^\dagger D(t)B,\qquad
 D(t)=\diag\bigl(e^{-i(k+1/2)t}\bigr)_{k=0}^{N-1}.
\end{equation}}

To fast-forward an input state by an arbitrary time $t$, we map the system to the energy basis via $B$, execute a diagonal matrix of phase rotations corresponding to the exact QHO energy eigenvalues $E_k=k+1/2$, and return to the spatial basis using $B^\dagger$. The exact finite-space time-evolution sequence is:
\begin{equation}
\begin{quantikz}
\lstick{$\ket{S_N\psi(0)}$} & \qwbundle{n} & \gate{B} & \gate{D(t)} & \gate{B^\dagger} & \qw \rstick{$\ket{S_N\psi(t)}$}
\end{quantikz}
\end{equation}
The diagonal phases factor over the binary digits of $k$, together with the global phase $e^{-it/2}$. Thus the number of phase rotations need not grow with the number of elapsed oscillator periods; their precision and classical angle calculation remain part of the implementation. Figure~\ref{fig:combined_demos}(a) illustrates the coherent-state trajectory. A nontrivial coherent state has infinitely many Hermite coefficients, so its use in $\hilbert_N$ requires truncation. The exact projected evolution preserves that truncation error in $L^2$; transform synthesis and phase synthesis contribute separate errors.

Both plots below use normalized projections onto $N=32$ modes and reconstruct the continuous density $|\psi(x)|^2$ from those modes. The omitted squared norms are approximately $1.46\times10^{-18}$ for the coherent state and $6.87\times10^{-17}$ for the cat state. The accompanying reproducibility package supplies the coefficients, plotted data, and generating code.

\begin{figure}[ht]
    \centering

    \begin{subfigure}[t]{0.48\linewidth}
        \centering
        \includegraphics[width=\linewidth,alt={Density-band plot of a coherent wavepacket oscillating sinusoidally in position. The dashed curve marks its center; parameters are in the caption.}]{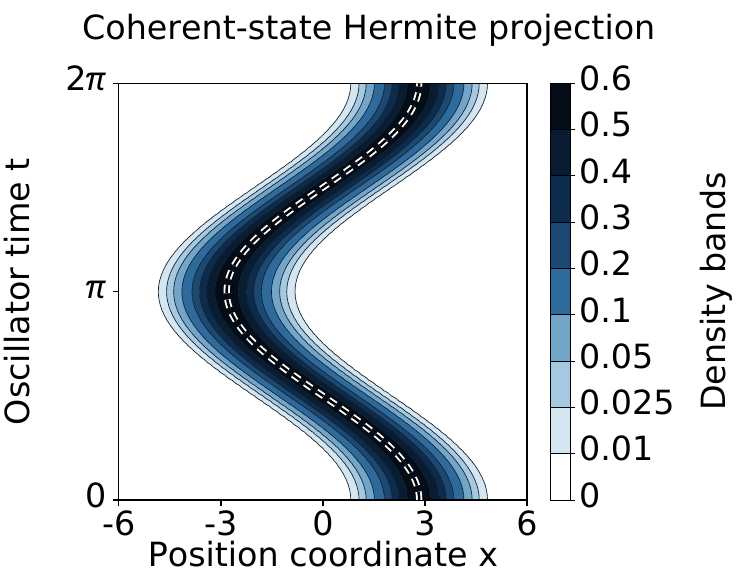}
        \caption{Coherent state with real amplitude $2$, evolved for $0\leq t\leq2\pi$ in oscillator units. The dashed curve is $x=2\sqrt2\cos t$.}
        \label{fig:qho_evolution}
    \end{subfigure}
    \hfill
    \begin{subfigure}[t]{0.48\linewidth}
        \centering
        \includegraphics[width=\linewidth,alt={Density-band plot of an even cat state under fractional Fourier rotation. Two position peaks become momentum fringes halfway through the rotation.}]{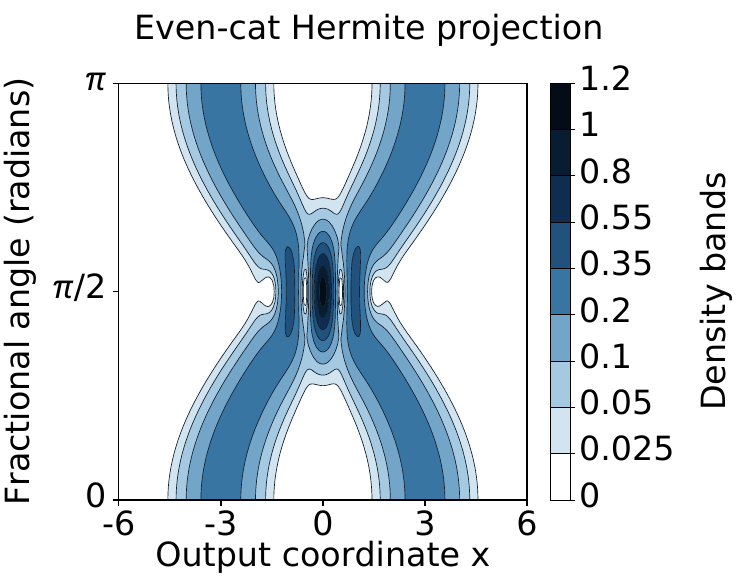}
        \caption{Even cat state with displacement $x_0=3$, rotated through $0\leq\alpha\leq\pi$. The momentum-domain fringes occur at $\alpha=\pi/2$; the even position profile returns at $\alpha=\pi$.}
        \label{fig:frft_demo}
    \end{subfigure}

    \caption{Visual demonstrations of the Quantum Hermite Transform applied to physical time-evolution (left) and continuous-variable signal processing (right). Filled bands show probability-density intervals. The finite-transform identities apply to the retained Hermite modes; these illustrations are not gate-level or hardware error benchmarks.}
    \label{fig:combined_demos}
\end{figure}

\subsection{Continuous Fractional Fourier Analysis}
In signal processing, the standard quantum Fourier transform acts on a cyclic group, whereas the Hermite representation describes functions on the real line. We use the continuous Fourier convention $(F\psi)(p)=(2\pi)^{-1/2}\int_{\mathbb R}e^{-ipx}\psi(x)\,dx$. Its fractional extension is defined spectrally by $F_\alpha\varphi_k=e^{-ik\alpha}\varphi_k$, with $F_0=\id$ and $F_{\pi/2}=F$~\cite[Secs.~1--3]{kutay2002frft}.

The fractional Fourier transform (FrFT) rotates a signal in phase space by an angle $\alpha$. On $\hilbert_N$, its exact representation on weighted samples is $B^\dagger\diag(e^{-ik\alpha})_{k=0}^{N-1}B$. Writing $\psi_\alpha=F_\alpha\psi_0$, this gives the circuit
\begin{equation}
\begin{quantikz}
\lstick{$\ket{S_N\psi_0}$} & \qwbundle{n} & \gate{B} & \gate{\diag(e^{-ik\alpha})} & \gate{B^\dagger} & \qw \rstick{$\ket{S_N\psi_\alpha}$}
\end{quantikz}
\end{equation}
The QHO evolution differs from this transform by the global phase $e^{-it/2}$ when $\alpha=t$~\cite[Eq.~(16)]{kutay2002frft}. Both operators preserve $\hilbert_N$ and the norm.

The example in Figure~\ref{fig:combined_demos}(b) uses a bimodal Schr\"odinger cat state. In the position domain ($\alpha=0$), the state is described by a superposition of displaced coherent states, with real displacement $x_0$:
\begin{equation}
\psi_0(x) \propto e^{-\frac{1}{2}(x-x_0)^2} + e^{-\frac{1}{2}(x+x_0)^2}.
\end{equation}
As the signal is rotated into the momentum domain ($\alpha=\pi/2$), the Fourier transform of these displaced spatial peaks yields the interference pattern
\begin{equation}
\tilde \psi(p) = \psi_{\pi/2}(p) \propto e^{-\frac{1}{2}p^2} \cos(p x_0).
\end{equation}
These are continuous formulas. Their normalized Hermite projections yield finite-register examples, with the discarded tail and any aliasing from the sampled tail treated separately. The finite operator preserves the projected norm exactly and imposes no periodic spatial boundary.

Supplementary Section~\ref{S-app:noise-diagnostic} gives a small, reproducible illustration of qubit noise viewed through Hermite-mode populations. The transform specifies the preparation and measurement bases for that observable. We pose this mode-population diagnostic as an open direction for future research to further formalize.

\section{Averaging Operator \texorpdfstring{$R$}{R} on a Uniform Grid}\label{sec:adaptive-windows}

To implement the weighted-root sampling map $B^\dagger$ in \eqref{eq:weighted-input}, we average uniform-grid samples near each Hermite root.  Polynomial reproduction controls the error on the retained Hermite space; radii at the local root-spacing scale bound the averaging norm uniformly in $N$ for fixed kernel degree and hence control coherent amplification. We prove Theorem~\ref{thm:near-linear} using the Introduction's notation and error conventions.

While coherent averaging and reversible changes between uniform sampling grids were pioneered by Kitaev and Webb \cite{kitaev2009resampling}, their construction does not map to non-uniform, weighted quadratures. Here, we extend and realize the standard polynomial kernel and amplitude amplification specifically to solve the non-uniform weighted Gauss--Hermite sampling problem.

\subsection{Local weights, separation, and analytic control}

The averaging row norm depends on the effective weight divided by its
radius, so we need a local weight estimate.  Define
\begin{equation}\label{eq:aw-local-scales}
 v_i=w_i e^{x_i^2},\qquad
 q_i=2N-x_i^2+N^{1/3},\qquad
 c=\frac1{64},\qquad r_i=\frac{c}{\sqrt{q_i}},
 \qquad 1\leq i\leq N.
\end{equation}
Here, $v_i$ is the effective weight, $q_i$ defines the local density scale tracking spatial root compression, $c$ is a geometric safety constant guaranteeing disjoint windows, and $r_i$ is the averaging radius at this local scale.
In particular $(B^\dagger)_{i,k}=\sqrt{v_i}\varphi_k(x_i)$, with the
root-label convention above.  A radius used below is always positive.

\begin{lemma}[Local weight and root separation]\label{lem:local-weight}
There are absolute constants $c_W,C_W>0$ such that
\begin{equation}\label{eq:aw-local-weight}
 N^{1/3}\leq q_i\leq3N,\qquad
 \frac{c_W}{\sqrt{q_i}}\leq v_i\leq\frac{C_W}{\sqrt{q_i}}.
\end{equation}
Every consecutive root gap $g_i=x_{i+1}-x_i$ satisfies
\begin{equation}\label{eq:aw-local-gap}
 g_i\geq\max\{q_i^{-1/2},q_{i+1}^{-1/2}\},\qquad
 r_i+r_{i+1}\leq2c g_i.
\end{equation}
\end{lemma}
The proof is in Supplementary Section~\ref{S-app:local-weight}.

For parameter selection, one may use certified root and weight intervals
to obtain a dyadic number
\begin{equation}\label{eq:aw-certified-weight}
 \overline C_W\geq\max\{1,\max_i v_i\sqrt{q_i}\}.
\end{equation}
Refining the intervals makes this bound a fixed-factor estimate of the maximum.  Lemma~\ref{lem:local-weight} then bounds $\overline C_W$ uniformly without a numerical source constant.

Hermite functions are entire, but the approximation requires a bound
uniform in the degree and the root, including the edges.  The next lemma
provides that bound on a disk scaled to $r_i$.  Cauchy's estimate will then
control all Taylor remainders with one choice of the kernel degree.

\begin{lemma}[Analytic bound at the local scale]\label{lem:aw-mehler}
Let $c_i\in\mathbb R$. Suppose $|c_i-x_i|\leq h\leq r_i$ and $0<\rho_i\leq r_i$.
There is an absolute constant $C_M>0$ such that, for every root and every complex $z$ with $|z|\leq2$,
\begin{equation}\label{eq:aw-mehler-row}
 \sum_{k=0}^{N-1}|\varphi_k(c_i+\rho_i z)|^2 \leq C_M\sqrt{q_i}.
\end{equation}
Consequently, there exists an absolute constant $B_0 > 0$ such that the row vector
$f_i(z)=(\sqrt{v_i}\varphi_k(c_i+\rho_i z))_{k=0}^{N-1}$
has norm at most $B_0$. The analytic matrix with rows $f_i(z)$ has operator norm at most
$B_0\sqrt N$ on this disk.
\end{lemma}
The proof is in Supplementary Section~\ref{S-app:mehler}.

\subsection{The Averaging Operator \texorpdfstring{$R$}{R}}

A reproducing kernel turns uniform-grid samples into root samples by
reproducing the initial Taylor terms; midpoint quadrature controls the
discretization error. The Legendre reproducing kernel supplies the required polynomial weights.

Let $P_j$ be the Legendre polynomial with
$P_j(1)=1$.  For an integer $d\geq0$ define
\begin{equation}\label{eq:aw-kernel}
 K_d(u)=\sum_{j=0}^d(2j+1)P_j(0)P_j(u),
 \qquad C_0=(d+1)^2,\quad C_1=(d+1)^4.
\end{equation}
The orthogonal-polynomial reproducing-kernel identity specializes to
\begin{equation}\label{eq:aw-reproduction}
 \frac12\int_{-1}^1K_d(u)p(u)\,du=p(0),
 \qquad \deg p\leq d.
\end{equation}
Indeed, expand $p$ in the Legendre basis and use
$\int_{-1}^1P_jP_k=2\delta_{j,k}/(2j+1)$.
The bounds $|P_j(u)|\leq1$ and
$|P_j'(u)|\leq j(j+1)/2$ on $[-1,1]$ give
\begin{equation}\label{eq:aw-kernel-bounds}
 \|K_d\|_\infty\leq C_0,\qquad
 \|K_d'\|_\infty\leq C_1.
\end{equation}

Let $L$ be a power of two, put $h=\sqrt{2\pi/L}$, and define the $L\times N$ uniform sampling matrix by
\begin{equation}\label{eq:uni_grid_def}
 I_L=[-L/2,L/2),\qquad y_\ell=h\ell,
 \quad \ell\in I_L\cap\mathbb Z,\qquad
 J_L[\ell,k]=\sqrt h\,\varphi_k(y_\ell).
\end{equation}
A signed grid address $\ell$ is stored as its $\log_2L$-bit residue, interpreted in $[-L/2,L/2)$. All address arithmetic and interval tests below use this convention.
Assuming the quantum grid has been made sufficiently fine, for example $h\leq\min_i r_i/16$, choose an even power of two $m_i$ and an integer start address $a_i$ such that
\begin{equation}\label{eq:aw-window-choice}
 \rho_i=\frac{m_i h}{2}\in(r_i/4,r_i],\qquad
 c_i=h\left(a_i+\frac{m_i-1}{2}\right),\qquad
 |c_i-x_i|\leq h
\end{equation}
where $\rho_i$ is the averaging window radius, and $c_i$ is the window midpoint which must be chosen to be sufficiently close to the $i$-th root $x_i$. For each root $i$ let the $m_i$ coordinates inside the $i$-th window be
\begin{equation}\label{eq:aw-midpoints}
 u_{i,j}=\frac{2j+1-m_i}{m_i},\qquad 0\leq j<m_i.
\end{equation}
Then $c_i+\rho_i u_{i,j}=h(a_i+j)$ exactly.  The physical cells
$[c_i-\rho_i,c_i+\rho_i]$ are disjoint
by Lemma~\ref{lem:local-weight}. Indeed, for an adjacent gap $g_i$, the two radii and center displacements total at most $(2c+c/8)g_i<g_i$. The grid will also be chosen large enough to contain every window.

The continuous average and its midpoint discretization are, respectively,
\begin{align}
 T_{\mathrm{cont}}[i,k]
 &=\frac{\sqrt{v_i}}2\int_{-1}^1
     K_d(u)\varphi_k(c_i+\rho_i u)\,du,
       \label{eq:aw-continuous-window}\\
 T[i,k]
 &=\frac{\sqrt{v_i}}{m_i}\sum_{j=0}^{m_i-1}
     K_d(u_{i,j})\varphi_k(c_i+\rho_i u_{i,j}).
       \label{eq:aw-discrete-window}
\end{align}
The matrix $T$ approximates $B^\dagger$; the following proposition bounds its error before circuit synthesis.

\begin{proposition}[Uniform approximation by adaptive windows]
\label{prop:aw-approximation}
Under \eqref{eq:aw-window-choice}, define
\begin{equation}\label{eq:aw-error-coefficients}
 \alpha_d=B_0C_0\sqrt N=O(\sqrt N(d+1)^2),\qquad
 \alpha_h=\frac{\sqrt3B_0}{c}(1+C_0+C_1)N.
\end{equation}
Then
\begin{equation}\label{eq:aw-error}
 \|T-B^\dagger\|\leq E_d(h):=\alpha_d2^{-d}+\alpha_hh.
\end{equation}
Both coefficients depend on $N$ and $d$, through
$C_0=(d+1)^2$ and $C_1=(d+1)^4$, but are independent of $h$ and $L$.
\end{proposition}

Supplementary Section~\ref{S-app:approximation} bounds the three errors separately.
The coefficients satisfy $\alpha_d=O(\sqrt N(d+1)^2)$ and
$\alpha_h=O(N(d+1)^4)$, so logarithmic degree and inverse-polynomial
grid spacing suffice.

To form signed averages of uniform-grid samples, define the rectangular matrix $R:\mathbb{C}^L\to\mathbb{C}^N$ by the following entries on each window, with all other entries zero:
\begin{equation}\label{eq:aw-rectangular}
 R[i,a_i+j] = \frac{\sqrt{v_i}\,K_d(u_{i,j})}{m_i\sqrt{h}}, \qquad D_i = \frac{1}{m_i}\sum_{j=0}^{m_i-1}K_d(u_{i,j})^2,
\end{equation}
where $D_i$ serves as the discrete normalization constant for the $i$-th root.

Recall that the grid addresses map exactly to our local coordinates via $y_{a_i+j} = c_i + \rho_iu_{i,j}$. This alignment allows $R$, when applied to the uniform grid sampler $J_L$, to perfectly reproduce our discrete approximation matrix $T$:
\begin{align}
    (R J_L)[i,k] &= \sum_{j=0}^{m_i-1} \left( \frac{\sqrt{v_i}K_d(u_{i,j})}{m_i\sqrt{h}} \right) \left( \sqrt{h}\varphi_k(y_{a_i+j}) \right) \\
    &= \frac{\sqrt{v_i}}{m_i} \sum_{j=0}^{m_i-1} K_d(u_{i,j}) \varphi_k(c_i + \rho_i u_{i,j}) \\
    &= T[i,k]\label{eq:aw-sampling-factorization}
\end{align}

Because the physical windows do not overlap, the rows of $R$ are strictly orthogonal. This allows for a direct formula for the overall operator norm:
\begin{equation}\label{eq:aw-row-normalization}
 RR^\dagger = \mathrm{diag}\left(\frac{v_iD_i}{2\rho_i}\right)_{i=1}^N, \qquad \|R\|^2 = \max_i\frac{v_iD_i}{2\rho_i}.
\end{equation}

The integral norm of the kernel gives a sharper normalization than its
pointwise bound. Orthogonality yields
\begin{equation}\label{eq:aw-kernel-energy}
 \begin{aligned}
 \frac12\int_{-1}^1K_d(u)^2\,du&=K_d(0),\\
 1\leq K_d(0)&\leq d+1.
 \end{aligned}
\end{equation}
To see the upper bound, use Legendre parity and
$P_{2r}(0)=(-1)^r\binom{2r}{r}/4^r$. Telescoping gives
$K_d(0)=(2r+1)^2\binom{2r}{r}^2/16^r$ for $r=\lfloor d/2\rfloor$.
The ratio of successive factors $\binom{2r}{r}/4^r$ is
$(2r-1)/(2r)$, so induction gives
$\binom{2r}{r}^2/16^r\leq1/(2r+1)$; the lower bound is the
$j=0$ term.

Since $K_d^2$ has derivative bounded by $2C_0C_1$, midpoint quadrature gives
$|D_i-K_d(0)|\leq C_0C_1/m_i\leq1/4$ whenever $m_i\geq4C_0C_1$.
Combining this estimate with $\rho_i>r_i/4$ and
$v_i\sqrt{q_i}\leq\overline C_W$ gives
\begin{equation}\label{eq:aw-discrete-norm}
 \begin{aligned}
 \frac34 K_d(0)&\leq D_i\leq d+\frac54,\\
 \|R\|^2&\leq\frac{2\overline C_W}{c}\left(d+\frac54\right).
 \end{aligned}
\end{equation}
Choose $\overline A$ as the least power of two satisfying
\begin{equation}\label{eq:aw-Abar}
 \begin{gathered}
 \overline A\geq\max\left\{2,
  2\sqrt{\frac{2\overline C_W}{c}\left(d+\frac54\right)}\right\},
 \\
 \overline A=O(\sqrt{d+1}).
 \end{gathered}
\end{equation}
This choice precedes the final grid size and ensures
$\overline A\geq2\|R\|$ on every admissible finer grid.

\subsection{Realization of the Averaging Windows Operator \texorpdfstring{$R$}{R}}

To realize $R$, prepare its normalized row states with equal preliminary success amplitudes, permitting a common amplification sequence.

For the certified disjoint dyadic windows above, define
\begin{equation}\label{eq:aw-window-state}
 \ket{\chi_i}=\sum_{j=0}^{m_i-1}
       \frac{K_d(u_{i,j})}{\sqrt{m_iD_i}}\ket j,
 \qquad
 W\ket{i-1}=\sum_{j=0}^{m_i-1}
       \frac{K_d(u_{i,j})}{\sqrt{m_iD_i}}\ket{a_i+j}.
\end{equation}

Disjoint supports and the definition of $D_i$ make $W$ an isometry.
We construct its root-to-grid extension and use the reversed circuit
to average grid samples.  The construction below prepares the signed
states, unloads their metadata, and erases the root label coherently.

\begin{lemma}[Circuit implementation of the window isometry]
\label{lem:aw-row-state}
$W$ is an isometry with a reversible unitary extension.  For
$0<\zeta<1/2$, its extension and its inverse can be approximated to
operator error at most $\zeta$ using
\[
 O\bigl(N\operatorname{poly}(d,\log N,\log L,\log(1/\zeta))\bigr)
\]
logical gates and
$\operatorname{poly}(d,\log N,\log L,\log(1/\zeta))$ quantum width.
The same bounds hold for a unitary whose designated input-output block
is $R/\overline A$.
\end{lemma}

\begin{proof}
We start by showing a staged construction of the operator $W$, detailing how the system prepares local states, amplifies them, and reversibly uncomputes auxiliary data. To clarify the internal mechanics, each stage below tracks the exact evolution of the quantum state for a simplified running example: a single input root labeled $\ket{i-1}_{\text{root}}$, accompanied by initially empty auxiliary registers for metadata, local indices, the physical grid, and a success flag.
\begin{enumerate}
    \item \textbf{Metadata Loading}
    \begin{enumerate}
        \item Define $m_i = 2^{s_i}$ and allocate a shared quantum register of size $s_{\max} = \max_i s_i$ to hold the local index $j$.
        \item Load the hardwired starting address $a_i$, length $s_i$, and normalization constant $D_i$ by sequentially scanning the $N$ classical records. Each scan computes equality to a label, conditionally XORs its record, and clears the comparison flag; its $O(N)$ times polynomial record-length gate cost is included.

        \item \textbf{Output State:} The root register acts as a control to pull classical data into the active metadata registers.
        \begin{equation}
            \ket{i-1}_{\mathrm{root}}\ket{a_i,s_i,D_i}_{\mathrm{meta}}
            \ket0_{\mathrm{local}}\ket0_{\mathrm{grid}}\ket0_{\mathrm{flag}}
        \end{equation}
    \end{enumerate}
 \item \textbf{Local State Preparation}
    \begin{enumerate}
        \item Apply Hadamard gates to precisely the lowest $s_i$ bits, leaving the higher bits zero, to generate a uniform superposition over valid local addresses $0 \leq j < m_i$.
        \item Reversibly compute the local weight, rotate an auxiliary flag qubit to encode the target amplitude
        \begin{equation}\label{eq:aw-signed-amplitude}
            f_i(j)=\frac{bK_d(u_{i,j})}{\sqrt{D_i}}
        \end{equation}
        and immediately uncompute the arithmetic to clear the workspace. Set
        $t_0=4(d+1)$ and $b=\sin(\frac{\pi}{4t_0+2})$.
        Cauchy--Schwarz gives
        $|K_d(u)|^2\leq K_d(0)\sum_{j=0}^d(2j+1)P_j(u)^2\leq K_d(0)(d+1)^2$.
        Thus \eqref{eq:aw-discrete-norm} and $b<1/[4(d+1)]$ imply
        $|f_i(j)|<1/(2\sqrt3)<1/2$, leaving fixed slack for the rotation.
        For invalid addresses or malformed metadata define the amplitude to be zero; clamp it to $[-1/2,1/2]$ on the remaining register states. These reversible bounded extensions agree with the specified amplitudes on loaded valid records and define a uniformly conditioned rotation on the full register space.
        \item \textbf{Output State:} The local register is placed in a uniform superposition, with total success probability $b^2$.
      \begin{equation}
          \begin{aligned}
          &\ket{i-1}_{\mathrm{root}}\ket{a_i,s_i,D_i}_{\mathrm{meta}}
          \ket0_{\mathrm{grid}}\\
          &\quad\otimes\frac{1}{\sqrt{m_i}}\sum_{j=0}^{m_i-1}
          \ket j_{\mathrm{local}}
          \left(f_i(j)\ket{\mathrm{good}}_{\mathrm{flag}}
          +\sqrt{1-f_i(j)^2}\ket{\mathrm{bad}}_{\mathrm{flag}}\right).
          \end{aligned}
      \end{equation}
    \end{enumerate}

    \item \textbf{Reduction to Exact Grover Search}
    \begin{enumerate}
        \item Let $A$ be the row-controlled state-preparation circuit, and use the Grover-type sequence
        {\begin{equation}\label{eq:aw-inner-amplification}
          (-A R_0 A^\dagger R_g)^{t_0}A.
        \end{equation}}
        Here $R_g=\id-2\Pi_g$ flips the good flag and $R_0=\id-2\Pi_0$ reflects about zero preparation registers, leaving the root label and loaded metadata untouched.
        \item Because the initial amplitude was artificially shrunk by the specific parameter $b$, this becomes an \textit{exact} Grover search. Indeed, every row has good component $b\ket{\chi_i}\ket{\mathrm{good}}$. The common angle $\theta_0=\pi/(4t_0+2)$ rotates to $(2t_0+1)\theta_0=\pi/2$ with the same positive phase in every row \cite[Sec.~2.1]{brassard2000}. The scalar minus in the reflection product fixes this phase.
        \item It perfectly isolates the target state $|\chi_i\rangle$ in the ideal circuit. Arithmetic registers are zero because the preparation explicitly uncomputes them; the metadata remains loaded. This direct sum of equal-angle rotations also works on root-label superpositions.
        \item \textbf{Output State:} The target state $\vert{}\chi_i\rangle$ is cleanly isolated in the local register, the arithmetic memory is empty, and the success flag achieves exactly $100\%$ probability.
        \begin{equation}
            \ket{i-1}_{\mathrm{root}}\ket{a_i,s_i,D_i}_{\mathrm{meta}}
            \ket{\chi_i}_{\mathrm{local}}\ket0_{\mathrm{grid}}
            \ket{\mathrm{good}}_{\mathrm{flag}}
        \end{equation}
    \end{enumerate}

    \item \textbf{Address Mapping and Reversible Erasure}
    \begin{enumerate}
        \item Shift the local index $j$ into an absolute grid position by computing $\ell=a_i+j$ into a new output register. XOR the computed value $\ell-a_i$ into the local register and reverse the subtraction arithmetic. Then reverse the metadata-loading scan while the controlling root label is still present. The local and metadata registers are now both zero.
        \item To safely discard the initial root label $i-1$ without collapsing the quantum superposition, exploit the fact that the physical windows are strictly disjoint. Evaluate a piecewise reverse-lookup function $\lambda(\ell) = \lambda(a_i+j) = i-1$ via a sequential scan, which uniquely identifies the root associated with any valid grid position. Set $\lambda(\ell)=0$ outside all windows. Each interval comparison conditionally XORs its hardwired label and uncomputes its test, so this defines a reversible operation on every grid input.
        \item XOR this computed reverse-lookup value against the original input register. Because $\lambda(\ell)$ exactly equals $i-1$, this cleanly subtracts the label from itself, restoring the root register to $|0\rangle$ and isolating the final grid amplitudes.
        \item \textbf{Output State:} The metadata, local index, and initial root label are perfectly uncomputed, transferring all window-state amplitudes solely into the active grid register.
        \begin{equation}
            \ket0_{\mathrm{root}}\ket0_{\mathrm{meta}}\ket0_{\mathrm{local}}
            \otimes\left(\sum_{j=0}^{m_i-1}
            \frac{K_d(u_{i,j})}{\sqrt{m_iD_i}}\ket{a_i+j}_{\mathrm{grid}}\right)
            \ket{\mathrm{good}}_{\mathrm{flag}}
        \end{equation}
    \end{enumerate}

\end{enumerate}

Toggle the now fixed good flag to zero. Each address operation uses $O(\log L)$ bits; comparisons and arithmetic reuse their workspace. Reversing this complete circuit gives the inverse of the clean extension of $W$.

Having successfully constructed the isometry $W$, the circuit proceeds to use it to synthesize the complete adjoint averaging matrix $R^\dagger$. The scaled operator $R^\dagger/\overline A$ is a contraction, so it can be embedded as a block of a larger unitary. The unscaled $R^\dagger$ need not shrink amplitudes.

We achieve this by algebraically decomposing the target matrix into two distinct operations:
{\begin{equation}\label{eq:aw-R-factorization}
 R^\dagger=W\operatorname{diag}(\lambda_i),\qquad
 \lambda_i=\sqrt{\frac{v_iD_i}{2\rho_i}}.
\end{equation}}
This decomposition is proven in the supplemental Sec.~\ref{S-sec:R_decomposition}.

\begin{itemize}
\item \textbf{Step 1: The Diagonal Scaling ($\mathrm{diag}(\lambda_i)$):}
On root label $i-1$, a separate flag rotation has success amplitude $\lambda_i/\overline A\leq1/2$. Load the row data, compute and apply the rotation, reverse its arithmetic, and unload the data while the root label is still present. Assign success to the zero flag state. The rotation angle is $2\arcsin(\lambda_i/\overline A)$ in the usual $R_y(\theta)=e^{-i\theta Y/2}$ convention, with a fixed flag relabeling if needed.
\item \textbf{Step 2: The Spatial Mapping ($W$):}
With the local quadrature scaling successfully encoded into the flag qubit's probability amplitude, the circuit executes the exact forward sequence of $W$ (Stages 1 through 4) to spatially distribute the root state across the physical grid. The separate row flag is left untouched by this sequence, including its inner amplification.
\end{itemize}

Through block encoding, the combined effect of scaling the flag and applying $W$ successfully embeds the target matrix into the overall quantum state. The block from root inputs with zero workspace to grid outputs with zero workspace and successful row flag is $R^\dagger/\overline A$. Reversing the complete circuit and interchanging these designated subspaces gives $R/\overline A$ on arbitrary grid inputs.

All flag amplitudes use bounded full-space extensions, as in the local preparation, and stay away from $\pm1$. Thus amplitude error $\xi$ gives angle error $O(\xi)$. Polynomial evaluation, reciprocal square roots, and bounded-angle evaluation require polynomially many bits in $d,\log N,\log L,\log(1/\zeta)$. In particular, the Legendre recurrence bounds the absolute coefficient sum of $P_j$ by $3^j$, so that of $K_d$ is at most $C_0\,3^d$; logarithmically many guard bits in this bound suffice. Reversible arithmetic and fixed-phase Clifford+$T$ synthesis \cite{ross2016}, with error allocated per gate occurrence, approximate the full unitary. Every inverse call reverses the actual compiled circuit. Telescoping includes the $O(d+1)$ inner amplification calls and all output registers.

This yields a final logical gate count and total quantum width of
{\[
 O\bigl(N\operatorname{poly}(d,\log N,\log L,\log(1/\zeta))\bigr)
 \quad\text{and}\quad
 \operatorname{poly}(d,\log N,\log L,\log(1/\zeta)),
\]}
respectively. A rigorous bounding of the polynomial evaluation, arithmetic overhead, and unitary synthesis complexity is provided in Supplementary Section~\ref{S-app:finite_precision}.

\end{proof}

\subsection{Complete amplification and choice of precision}

Clean implementation also generates amplitudes outside the designated
output subspace.  The following two-reflection amplification estimate
controls the full output with fixed phase; for projected robust
amplification, see \cite{gilyen2018}.

The input and output registers can have different layouts, so we use separate clean embeddings. Padding with zero registers and fixed register permutations put both layouts in one circuit space.

\begin{lemma}[Fixed-phase amplification with all output registers]
\label{lem:clean-amplification}
Let $Q$ be an $N\times N$ unitary. Let $U_0$ be a unitary with designated rank-$N$ input and output embeddings $E_{\mathrm{in}}$ and $E_{\mathrm{out}}$ satisfying
{\begin{equation}
 E_{\mathrm{out}}^\dagger U_0 E_{\mathrm{in}}
 =\widetilde T/A,\qquad
 A\geq2,\qquad \|\widetilde T-Q\|\leq\delta\leq1/24.
\end{equation}}
By utilizing $O(A)$ alternating applications of $U_0$ and $U_0^\dagger$, alongside reflections about the designated subspaces and one additional qubit, we can construct an amplified unitary $U_{\mathrm{amp}}$ satisfying
{\begin{equation}\label{eq:aw-full-amplification}
 \|U_{\mathrm{amp}}E_{\mathrm{in}}-E_{\mathrm{out}}Q\|
 \leq(\pi+2)\delta.
\end{equation}}
The embeddings in the conclusion include the extra qubit in its designated fixed state.
\end{lemma}

The full-output and fixed-phase estimate, including the component outside the clean output subspace, is proved in Supplementary Section~\ref{S-app:amplification}. It gives a coherent implementation without postselection.

To supply the input data for our amplified transform, we use the uniform-grid sampler of Supplementary Theorem~\ref{S-thm:uniform-grid}. Based on the grid definitions in Eq.~\eqref{eq:uni_grid_def}, for any target preparation error $0<\tau\leq1/10$ and every power-of-two grid dimension $L\geq C(N+1)^3\tau^{-8}$ (where $C$ is an absolute constant), this subroutine provides a coherent sampling unitary $V_L$ satisfying
\begin{equation}\label{eq:aw-grid-interface}
\|V_L E_{\mathrm{in}}-E_{\mathrm{grid}}J_L\|\leq\tau.
\end{equation}
Here, $E_{\mathrm{grid}}$ embeds the $L$-dimensional sample register with all sampler work registers zero. This operation is efficient---requiring a circuit cost and width polynomial in $\log N$, $\log L$, and $\log(1/\tau)$---and maintains fixed column phases.
The norm includes every output register; it therefore allows small residual workspace amplitudes. The tolerance also accounts for $J_L$ not being exactly isometric. The theorem permits a larger dyadic grid, allowing us to satisfy the window-resolution conditions simultaneously.

The remaining choices separate naturally into two tasks: make the analytic and sampling errors small, then compile the resulting amplified circuit.

\begin{lemma}[Compatible parameter choice]\label{lem:param-choice}
For any $N\geq 2$ and $0<\varepsilon\leq 1/2$, there exist admissible choices for the polynomial degree $d$, grid size $L$, and precision tolerances such that all construction requirements are satisfied simultaneously. Specifically,
{\begin{gather}
 d=O(\log(N/\varepsilon)),\qquad
 \overline A=O(\sqrt{d+1}),\qquad
 L=\operatorname{poly}(N,1/\varepsilon),\label{eq:aw-compatible-size}\\
 E_d(h)\leq\varepsilon/24,\qquad
 0<\tau\leq1/10,\qquad
 \overline A\tau\leq\varepsilon/24.\label{eq:aw-compatible-errors}
\end{gather}}
The certified windows lie inside the grid, are pairwise disjoint, satisfy \eqref{eq:aw-window-choice} and $m_i\geq4C_0C_1$, and obey $\overline A\geq2\|R\|$. The normalization $\overline A$ is fixed independently of the final grid refinement.
\end{lemma}

\begin{lemma}[Finite-precision compilation]\label{lem:finite-compilation}
Use the parameters of Lemma~\ref{lem:param-choice}. Compose the actual sampler $V_L$ with the exact window block of Lemma~\ref{lem:aw-row-state}, and apply the amplification construction of Lemma~\ref{lem:clean-amplification} with $A=\overline A$. Compiling the window operations and additional coin rotation changes this complete circuit by at most $\varepsilon/2$ in full-unitary operator norm. The resulting Clifford+$T$ circuit uses
{\[
 O\bigl(N\operatorname{polylog}(N,1/\varepsilon)\bigr)
 \quad\text{logical gates and}\quad
 \operatorname{polylog}(N,1/\varepsilon)
 \quad\text{logical qubits}.
\]}
The sampler circuit is identical in the comparison and compiled sequences, and every inverse is the reverse of the actual forward circuit. Classical certification and circuit generation take polynomial time in $N,1/\varepsilon$ and are charged separately.
\end{lemma}

The proofs of these two lemmas are in Supplementary Sections~\ref{S-sec:param_choice} and \ref{S-sec:precision_compilation}.

\begin{proof}[Proof of Theorem~\ref{thm:near-linear}]
We begin by selecting compatible parameters according to Lemma~\ref{lem:param-choice}. We compose the uniform-grid unitary $V_L$ in \eqref{eq:aw-grid-interface} with the ideal block for $R/\overline A$ from Lemma~\ref{lem:aw-row-state}.
The comparison circuit uses exact window rotations but the actual sampler at accuracy $\tau$. Its good output subspace requires zero sampler and window workspace and a successful row flag. Write its block as $\widetilde T/\overline A$. On ideal grid samples this block would be $RJ_L/\overline A=T/\overline A$ by \eqref{eq:aw-sampling-factorization}. Unitary composition and projection cannot increase the sampler error, so
{\begin{equation}\label{eq:aw-good-error}
 \|\widetilde T-B^\dagger\|
 \leq E_d(h)+\overline A\tau=:\delta
 \leq\varepsilon/12\leq1/24.
\end{equation}}
Padding and fixed register permutations supply the two rank-$N$ clean embeddings. Apply Lemma~\ref{lem:clean-amplification} with $Q=B^\dagger$ and $A=\overline A$. The complete amplified output is within $(\pi+2)\delta<\varepsilon/2$ of $E_{\mathrm{out}}B^\dagger$. Lemma~\ref{lem:finite-compilation} contributes at most $\varepsilon/2$ further error, proving \eqref{eq:aw-main-contract} with the stated gate count and width. Tensoring with an identity includes any external reference system.

Finally, multiplying~\eqref{eq:aw-main-contract} on the left by $U_{N,\varepsilon}^\dagger$ and on the right by $B$ gives
{\begin{equation}
 \|U_{N,\varepsilon}^\dagger E_{\mathrm{out}}
       -E_{\mathrm{in}}B\|\leq\varepsilon.
\end{equation}}
This confirms the inverse transform assertion using identical circuit resources.
\end{proof}

\section{Comparison to Known Implementations of the Transform}\label{rc:comparison}
We compare Theorem~\ref{thm:near-linear} against existing complete-transform synthesis methods, targeting the transform $B$ with a clean-input operator-norm error bounded by $\varepsilon$.
Here $\widetilde O$ suppresses factors polynomial in $1+\log N+\log(1/\varepsilon)$.

For our baselines, we first consider the complete reflection construction by Pli\'s and Zak, which requires quadratic total gates with logarithmic precision overhead \cite{plis2025}. We also compare against the unitary synthesis framework of Low, Kliuchnikov, and Schaeffer \cite{low2024trading}. While their parameterization can achieve an $\widetilde{O}(N^{3/2})$ $T$-gate count by utilizing $\widetilde{O}(\sqrt{N})$ auxiliary width, the online coherent-lookup Clifford gates still give a quadratic $\widetilde{O}(N^2)$ total gate bound \cite{low2024trading}.

Here $\lambda$ follows Sections~2--3 of~\cite{low2024trading}: it counts $b$-bit lookup blocks and uses $b\lambda$ borrowed qubits, with $b=O(\log(N/\varepsilon))$. The lookup requires $O(bN)$ online Clifford gates per reflection-state preparation; the $N$-reflection synthesis therefore has a quadratic total-gate upper bound. Classical preprocessing is a separate cost.

Table~\ref{rc:table} summarizes these upper bounds. By exploiting adaptive windows, our approach reduces the total gate complexity to near-linear $\widetilde{O}(N)$ while maintaining strictly polylogarithmic width. For positive quantities, $f\asymp g$ means $cg\leq f\leq Cg$ for constants $c,C>0$ independent of the varying parameters.

\begin{table}[ht]
\centering
\begin{tabular}{@{}>{\raggedright\arraybackslash}p{0.36\linewidth}ccc@{}}
\toprule
Construction & Online $T$ gates & Online total gates & Total width\\
\midrule
Adaptive windows, Theorem~\ref{thm:near-linear}
 & $\widetilde O(N)$ & $\widetilde O(N)$ & polylog\\[3pt]
Low--Kliuchnikov--Schaeffer, $\lambda=1$
 & $\widetilde O(N^2)$ & $\widetilde O(N^2)$ & polylog\\[3pt]
Low--Kliuchnikov--Schaeffer, $\lambda\asymp\sqrt N$
 & $\widetilde O(N^{3/2})$ & $\widetilde O(N^2)$
 & $\widetilde O(\sqrt N)$\\
\bottomrule
\end{tabular}
\caption{Upper bounds for the complete finite transform. All auxiliary qubits, including borrowed dirty qubits, contribute to width.}
\label{rc:table}
\end{table}

\section*{Conclusion and further questions}
\phantomsection
\addcontentsline{toc}{section}{Conclusion and further questions}

By matching the averaging width to the local effective quadrature weight, we establish a circuit for the finite transform with near-linear gate count and polylogarithmic width. Our continuous-input analysis isolates the underlying circuit error from truncation and aliasing from the sampled tail.

The resource bounds are asymptotic upper bounds for this construction. Practical performance comparisons require a concrete compiler with certified numerical constants and gate counts.

The oscillator and fractional Fourier examples illustrate operations diagonal in the retained Hermite space. Using them for hardware characterization will additionally require calibrated transform, preparation and measurement errors, as well as a verified noise model. Extending adaptive averaging to other quadratures is a further question: the local weight, separation and analytic estimates must be established for each new family.

\bibliographystyle{plain}
\begingroup
\raggedright
\bibliography{references}
\endgroup
\clearpage
\section*{Supplementary proofs and additional results}
\phantomsection
\addcontentsline{toc}{section}{Supplementary proofs and additional results}
\setcounter{section}{0}
\setcounter{theorem}{0}
\setcounter{equation}{0}
\setcounter{figure}{0}
\setcounter{table}{0}
\setcounter{remark}{0}
\renewcommand{\thesection}{S\arabic{section}}
\renewcommand{\thetheorem}{S\arabic{theorem}}
\renewcommand{\theequation}{S\arabic{equation}}
\renewcommand{\thefigure}{S\arabic{figure}}
\renewcommand{\thetable}{S\arabic{table}}
\renewcommand{\theremark}{S\arabic{remark}}
\renewcommand{\theHsection}{S\arabic{section}}
\renewcommand{\theHtheorem}{S\arabic{theorem}}
\renewcommand{\theHequation}{S\arabic{equation}}
\renewcommand{\theHfigure}{S\arabic{figure}}
\renewcommand{\theHtable}{S\arabic{table}}
\renewcommand{\theHremark}{S\arabic{remark}}

These supplementary proofs establish the analytic estimates, averaging-operator factorization, and assembly lemmas used in the main text. They also give the uniform-grid sampling contract, continuous-input error bounds, and introduce a mode-population noise diagnostic that we pose as an open problem for future hardware characterization.

We use the article's Hermite normalization, ordered roots
$x_1<\cdots<x_N$ of $H_N$, weights $w_i$, and transform
$B_{k,i}=\sqrt{w_i}\,h_k(x_i)$, with
$h_k=H_k/(2^k k!\sqrt\pi)^{1/2}$ and
$\varphi_k=e^{-x^2/2}h_k$.  Degree labels are $0\leq k<N$ and
root $i$ is stored as $i-1$.  Norms of finite-dimensional maps are
operator norms.  Circuit error includes all workspace, restored to
zero in the ideal output, with fixed global phase and arbitrary
external reference systems, as in article
Equation~\eqref{eq:aw-main-contract}.

The uniform sampler in Section~\ref{S-qht:appendix} uses the
analytic Mehler identity proved in Section~\ref{S-app:mehler}.
Its proof does not depend on the adaptive circuit
whose construction uses this sampler.

\section{Supporting proofs for adaptive sampling}
\label{S-app:analytic-proofs}
This appendix supplies the local weight and averaging estimates, exact
normalization arithmetic, and full-output amplification bound used in
main-text Section~\ref{sec:adaptive-windows}.  The analytic estimates do not depend
on the circuit construction.

\subsection{Weights and root separation}\label{S-app:local-weight}
\begin{proof}[Proof of main-text Lemma~\ref{lem:local-weight}]
The roots of $H_N$ are the eigenvalues of the real symmetric tridiagonal
$N\times N$ matrix with zero diagonal and off-diagonal entries
$\sqrt{k/2}$, $1\leq k<N$, by the Hermite recurrence
\cite[Sec.~2]{golub1969}.  Every absolute row
sum is below $\sqrt{2N}$, so the eigenvalue row-sum bound gives
$|x_i|<\sqrt{2N}$ and the bounds on $q_i$.

Orthogonality of $B$ gives the effective weight as
\[
 v_i=\left(\sum_{k=0}^{N-1}\varphi_k(x_i)^2\right)^{-1}.
\]
For $W(x)=e^{-x^2/2}$, let
$\lambda_N(W^2,x)=(\sum_{k=0}^{N-1}h_k(x)^2)^{-1}$ be the
Christoffel function.  Thus $v_i=\lambda_N(W^2,x_i)/W(x_i)^2$.

Write $W=e^{-Q}$.  Its Mhaskar--Rakhmanov--Saff number $a_N>0$
is determined by
\[
 N=\frac2\pi\int_0^1
       \frac{a_NtQ'(a_Nt)}{\sqrt{1-t^2}}\,dt.
\]
We use the Christoffel estimate of Levin and Lubinsky
\cite[displayed estimate in the abstract]{levinlubinsky1992}:
for a fixed admissible Freud weight $W=e^{-Q}$ and fixed $A>0$,
\begin{equation}\label{S-eq:aw-christoffel-input}
 \frac{\lambda_N(W^2,x)}{W(x)^2}
 \asymp\frac{a_N}{N}
 \left(\max\{N^{-2/3},1-|x|/a_N\}\right)^{-1/2},
 \quad |x|\leq a_N(1+A N^{-2/3}).
\end{equation}
The comparison constants are uniform in $N,x$.  In the Gaussian case
$Q(x)=x^2/2$ is even, $Q'(x)>0$ for $x>0$, and
$(xQ'(x))'/Q'(x)=2$, satisfying the source hypotheses.  The
defining integral gives
\[
 N=\frac{a_N^2}{2};
 \qquad a_N=\sqrt{2N}.
\]
All roots therefore lie in the domain of
\eqref{S-eq:aw-christoffel-input}.  An indexing convention that includes
degree $N$ gives the same value at $x_i$, because $H_N(x_i)=0$.

To compare the source's edge scale with $q_i$, put
$u_i=|x_i|/\sqrt{2N}$ and $\delta_i=1-u_i$.  These measure normalized
position and distance from the edge.  Define
\begin{equation}\label{S-def:M_i}
 M_i=\max\{N^{1/3},N\delta_i\}.
\end{equation}
The source estimate gives $v_i\asymp\sqrt2/\sqrt{M_i}$.  Since
$0\leq u_i<1$,
\[
 q_i=N^{1/3}+2N\delta_i(1+u_i),\qquad
 M_i\leq q_i\leq5M_i.
\]
Hence $v_i\asymp q_i^{-1/2}$, proving
main-text Equation~\eqref{eq:aw-local-weight}.  The root row-sum bound above and the
differential-equation argument below do not use the zero-location
corollary addressed in the source's erratum
\cite{levinlubinsky1995erratum}.

To prove separation, fix one root $x$ of a consecutive gap of
length $g$ and write $q=2N-x^2+N^{1/3}$.  Assume $g<q^{-1/2}$, then for every coordinate $y$ in
the gap
\begin{equation}
 2N+1-y^2 \underset{|x-y| \le g}{\leq} 2N+1-x^2+2|x|g  \leq q+\frac{2\sqrt{2N}}{\sqrt q}
 \leq(1+2\sqrt2)q
\end{equation}
where $q^3\geq N$ was used in the last step.  The Hermite equation
\cite[Table~18.8.1, Hermite row]{dlmf} gives
$\varphi_N''+(2N+1-y^2)\varphi_N=0$.  Multiplication by
$\varphi_N$ and integration between its two zeros, followed by the
Dirichlet Poincar\'e inequality, gives
\begin{equation}
 \frac{\pi^2}{g^2}
 \leq\frac{\int|\varphi_N'|^2}{\int|\varphi_N|^2}
 =\frac{\int(2N+1-y^2)|\varphi_N|^2}{\int|\varphi_N|^2}
 \leq(1+2\sqrt2)q.
\end{equation}
The denominator is nonzero.  Since $g^2q<1$, this contradicts
$\pi^2>1+2\sqrt2$, so the assumption is false and $g\geq q^{-1/2}$. Applying the result at both endpoints proves
main-text Equation~\eqref{eq:aw-local-gap}.
\end{proof}

\subsection{The scaled complex Hermite bound}\label{S-app:mehler}
\begin{proof}[Proof of main-text Lemma~\ref{lem:aw-mehler}]
Mehler's formula \cite[Eq.~18.18.28]{dlmf}, in the normalization of
$\varphi_k$, gives
\begin{equation}\label{S-eq:aw-complex-mehler}
 \sum_{k=0}^\infty e^{-tk}|\varphi_k(X+iY)|^2
 =\frac{e^{-\tanh(t/2)X^2+\coth(t/2)Y^2}}
        {\sqrt{\pi(1-e^{-2t})}},\qquad t>0.
\end{equation}
Normal convergence extends the real polynomial identity to complex
arguments.  The generating function
$e^{2zu-u^2}=\sum_{k\geq0}H_k(z)u^k/k!$, Cauchy's estimate on
$|u|=\sqrt{k/2}$, and Stirling's upper bound imply, on every compact
set of complex $z$,
\[
 \frac{|H_k(z)|}{\sqrt{2^k k!}}
 \leq C k^{1/4}e^{\sqrt{2k}|z|},\qquad k\geq1.
\]
Thus the two-variable Mehler series converges normally for $|r|<1$.
Apply the identity theorem in each variable, substitute $X+iY$ and
$X-iY$ with $r=e^{-t}$, and include the Gaussian factors to obtain
\eqref{S-eq:aw-complex-mehler}, including the positive $Y^2$ term.

Set $t=q_i^{-1}\leq1$ and write
$c_i+\rho_i z=x_i+\Delta+iY$, with $X=x_i+\Delta$.
The hypotheses give
\begin{align}
    |\Delta| &= |(c_i-x_i) + \rho_i \operatorname{Re} (z)| \leq r_i + 2 r_i \le 3r_i = \frac{3c}{\sqrt{q_i}} \\
    |Y| &= |\rho_i \operatorname{Im} (z)| \le 2r_i = \frac{2c}{\sqrt{{q_i}}}
\end{align}

For $k<N$, $e^{-tk}\geq e^{-Nt}$.  Multiplying the first $N$
terms by $e^{Nt}$ and then including the remaining nonnegative terms gives

\begin{equation}\label{S-eq:analy_exp_bound}
    \sum_{k=0}^{N-1} |\varphi_k(c_i +\rho_iz)|^2 \le e^{Nt} \frac{e^{-\tanh(t/2)X^2+\coth(t/2)Y^2}}
        {\sqrt{\pi(1-e^{-2t})}}
\end{equation}

The inequalities
$u-u^3/3\leq\tanh u\leq u$ and
$\coth(t/2)\leq1+2/t\leq3/t$ imply
\begin{align*}
 e^{Nt-\tanh(t/2)(x_i+\Delta)^2+\coth(t/2)Y^2}
 &\leq e^{t(N-x_i^2/2)+t^3x_i^2/24
          +t|x_i\Delta|+3Y^2/t}\\
 &\leq e^{1/2+1/12+3\sqrt2c+12c^2}
\end{align*}
Here $q_i^3\geq N$ controls the second and third terms. Concavity gives
$1-e^{-2t}\geq(1-e^{-2})t = \frac{1-e^{-2}}{q_i}$ for $0\leq t = q_i^{-1}\leq 1$. Combining this with the exponential
bound~\eqref{S-eq:analy_exp_bound} yields main-text Equation~\eqref{eq:aw-mehler-row}, with
\begin{equation}
 C_M = \frac{e^{7/12+3\sqrt2c+12c^2}}{\sqrt{\pi(1-e^{-2})}}.
\end{equation}
Multiply by $v_i\leq\overline C_W/\sqrt{q_i}$ to obtain the row bound. The $\sqrt{q_i}$ terms
cancel, bounding the row vector norm by $\sqrt{\overline C_W C_M}\leq B_0$. Summing the squared row bounds across all $N$ rows bounds
the Frobenius norm by $B_0\sqrt N$.
\end{proof}

\subsection{The three averaging errors}\label{S-app:approximation}
\begin{proof}[Proof of article Proposition~\ref{prop:aw-approximation}]
Let $F(z)$ have rows $f_i(z)$ from main-text Lemma~\ref{lem:aw-mehler}.
The same dimensionless variable $z$ parametrizes every row, even though
their physical centers and radii differ.
Write $F(z)=\sum_{j\geq0}F_jz^j$.  These Taylor coefficients satisfy
$\|F_j\|\leq B_0\sqrt N\,2^{-j}$ by Cauchy's integral formula on
the radius-two disk.  Thus its degree-$d$ Taylor polynomial has uniform
error at most $B_0\sqrt N\,2^{-d}$ on $[-1,1]$.  Applying
main-text Equation~\eqref{eq:aw-reproduction} to that polynomial to approximate $F(0)$ and using
the bound on $K_d$ in main-text Equation~\eqref{eq:aw-kernel-bounds} gives
\begin{equation}
 \|T_{\mathrm{cont}}-F(0)\|
 \leq B_0 C_0\sqrt N\,2^{-d}.
\end{equation}

The row lengths $m_i$ may also differ.  For a real $u\in[-1,1]$ the
radius-one disk about $u$ lies in $|z|\leq2$; Cauchy's estimate gives
$\|f_i'(u)\|_2\leq B_0$.  Hence
$g_i(u)=K_d(u)f_i(u)$ is Lipschitz with constant
$\|g_i'(u)\|_2 = \|f_i'(u)K_d(u) + f_i(u)K_d'(u)\|_2 \leq B_0(C_0+C_1)$.  A vector-valued Lipschitz function of constant $H$
has normalized midpoint error at most $H/(2m)$: the integral of the
distance to the midpoint of a cell of length $2/m$ is $m^{-2}$, and
the integral is normalized by $1/2$.  Apply this separately to each row
and then use the Frobenius norm to obtain
\[
 \|T-T_{\mathrm{cont}}\|
 \leq\frac{B_0(C_0+C_1)}2\sqrt{\sum_i m_i^{-2}}.
\]
Finally apply main-text Lemma~\ref{lem:aw-mehler} with center $x_i$ and radius
$r_i$.  Its derivative bound along the segment from $x_i$ to $c_i$
gives a row error at most $B_0h/r_i$, and therefore
$\|F(0)-B^\dagger\|\leq(B_0h/c)\sqrt{\sum_iq_i}$.
Combining these three estimates gives the more detailed bound
\begin{equation}\label{S-eq:aw-rowwise-error}
 \|T-B^\dagger\|
 \leq\frac{B_0h}{c}\sqrt{\sum_iq_i}
      +\alpha_d2^{-d}
      +\frac{B_0(C_0+C_1)}2\sqrt{\sum_i m_i^{-2}}.
\end{equation}
Since $m_i>r_i/(2h)$ and $\sum_iq_i\leq3N^2$,
\[
 \sqrt{\sum_i m_i^{-2}}
 \leq\frac{2h}{c}\sqrt{\sum_iq_i}
 \leq\frac{2\sqrt3Nh}{c}.
\]
Thus the first and third terms in \eqref{S-eq:aw-rowwise-error} sum to
at most $\alpha_hh$, with the coefficients defined in
main-text Equation~\eqref{eq:aw-error-coefficients}.  This proves article
Equation~\eqref{eq:aw-error}.
\end{proof}

\subsection{Amplification with fixed phase and full output}
\label{S-app:amplification}
\begin{proof}[Proof of main-text Lemma~\ref{lem:clean-amplification}]
Put $t=\lceil A\rceil$, $\theta=\pi/(4t+2)$, and
$s=\sin\theta$.  Adjoin an independent coin with positive good
amplitude $\beta=As<1$; the inequality follows from
$As\leq\pi A/(4t+2)<1$.  Require this coin to be good in the output
subspace, choosing its good value to be zero.  For the resulting
unitary $U$, the good block is
$s\widetilde T$.  Write
$\Pi_{\mathrm{in}}=E_{\mathrm{in}}E_{\mathrm{in}}^\dagger$ and
$\Pi_{\mathrm{out}}=E_{\mathrm{out}}E_{\mathrm{out}}^\dagger$,
and let $R_{\mathrm{in}}=\id-2\Pi_{\mathrm{in}}$ and
$R_{\mathrm{out}}=\id-2\Pi_{\mathrm{out}}$.  The amplification is
\begin{equation}\label{S-eq:aw-outer-sequence}
 U_{\mathrm{amp}}=(-U R_{\mathrm{in}}U^\dagger R_{\mathrm{out}})^tU.
\end{equation}
The displayed scalar signs and the positive coin amplitude fix its
global phase.  The scalar $-\id$ has the exact Clifford realization
$\sigma_x\sigma_z\sigma_x\sigma_z=-\id$ on one qubit.

Take a singular value decomposition
$\widetilde T=\sum_j\sigma_j\ket{u_j}\bra{v_j}$.
Since $Q$ is unitary, $1-\delta\leq\sigma_j\leq1+\delta$.
Define $\alpha_j=\arcsin(s\sigma_j)$ and the normalized bad vector
\[
 \ket{b_j}=\frac{U E_{\mathrm{in}}\ket{v_j}
       -\sin\alpha_j E_{\mathrm{out}}\ket{u_j}}
             {\cos\alpha_j}.
\]
The denominator is positive since $s\leq1/2$ and $\sigma_j\leq1.1$.
The vectors are orthogonal to the output subspace.  Subtracting the
good-component Gram matrix from that of the unitary images gives
$\langle b_j,b_k\rangle=\delta_{j,k}$.  Thus the planes
$\operatorname{span}\{E_{\mathrm{out}}u_j,b_j\}$ are pairwise
orthogonal.  In the ordered basis $(E_{\mathrm{out}}u_j,b_j)$ the
reflection product in \eqref{S-eq:aw-outer-sequence} is
\[
 \begin{pmatrix}
  \cos(2\alpha_j)&\sin(2\alpha_j)\\
  -\sin(2\alpha_j)&\cos(2\alpha_j)
 \end{pmatrix}.
\]
It follows, including the bad component, that
\[
 U_{\mathrm{amp}} E_{\mathrm{in}}v_j
  =\sin((2t+1)\alpha_j)E_{\mathrm{out}}u_j
    +\cos((2t+1)\alpha_j)b_j.
\]
At $\sigma_j=1$ the angle is $(2t+1)\theta=\pi/2$.
On the interval of singular values in question the derivative of
$\arcsin(s\sigma)$ is at most $2s$, whence
\[
 |(2t+1)\alpha_j-\pi/2|
 \leq2(2t+1)s\delta\leq\pi\delta.
\]
The full output's distance from $E_{\mathrm{out}}u_j$ is twice the
sine of half the angular deviation, at most $\pi\delta$.
Orthogonality of the planes extends this bound to operator norm.
Writing $Q_{\mathrm{pol}}=\sum_j\ket{u_j}\bra{v_j}$ for the unitary
polar factor, we have
\[
 \|U_{\mathrm{amp}} E_{\mathrm{in}}-E_{\mathrm{out}}Q_{\mathrm{pol}}\|
 \leq\pi\delta,\qquad
 \|Q_{\mathrm{pol}}-\widetilde T\|
 =\max_j|1-\sigma_j|\leq\delta.
\]
Together with $\|\widetilde T-Q\|\leq\delta$, this proves
main-text Equation~\eqref{eq:aw-full-amplification} for the unprojected output.  Tensoring
with an identity preserves operator norm and includes an external reference.
\end{proof}

\section{Factorization of the averaging operator}\label{S-sec:R_decomposition}
\begin{proof}
The factorization in main-text Equation~\eqref{eq:aw-R-factorization}
follows entrywise. Recall the normalized rows and scales from article
Lemma~\ref{lem:aw-row-state}:

\begin{equation}
    W|i-1\rangle=\sum_{j=0}^{m_i-1}\frac{K_d(u_{i,j})}{\sqrt{m_iD_i}}|a_i+j\rangle.
\end{equation}
All coefficients are real, so
\begin{equation}
    (W^\dagger)_{i,a_i+j}=\frac{K_d(u_{i,j})}{\sqrt{m_iD_i}}.
\end{equation}

Since $2\rho_i=m_ih$, the row scale is
\begin{equation}
    \lambda_i=\sqrt{\frac{v_iD_i}{2\rho_i}}=\sqrt{\frac{v_iD_i}{m_ih}}.
\end{equation}

Consequently
\begin{align}
    \lambda_i(W^\dagger)_{i,a_i+j} &= \sqrt{\frac{v_iD_i}{m_ih}}\left(\frac{K_d(u_{i,j})}{\sqrt{m_iD_i}}\right) \nonumber \\
    &= \frac{\sqrt{v_i}K_d(u_{i,j})}{m_i\sqrt{h}}
     =R[i,a_i+j].
\end{align}

Both sides vanish outside the row window, by article
Equation~\eqref{eq:aw-rectangular}. Thus
$R=\diag(\lambda_i)W^\dagger$ and
$R^\dagger=W\diag(\lambda_i)$.

\end{proof}

\section{Finite-Precision Synthesis and Complexity Analysis}
\label{S-app:finite_precision}

To approximate the ideal circuit, we must first control the arithmetic error in each flag amplitude. The separation of the amplitudes from $\pm1$ bounds the derivative of the $\arcsin$ function, ensuring that an arithmetic amplitude error $\xi$ translates to a bounded rotation error $O(\xi)$. Rational polynomial evaluation, reciprocal square roots on $D_i\geq1/2$, and bounded-angle functions admit arithmetic circuits polynomial in $d$, $\log L$, and the required output precision.

For example, let $L_j$ be the sum of the absolute monomial coefficients of the polynomial $P_j$. The Legendre recurrence gives $L_{j+1}\leq2L_j+L_{j-1}$, and hence $L_j\leq3^j$. Thus, the coefficient sum of the kernel $K_d$ is at most $C_0 3^d$. Horner evaluation on $|u|\leq1$ with absolute error $\xi$ uses $O(d+\log(d+1)+\log(1/\xi))$ magnitude, guard, and fractional bits, up to fixed arithmetic overhead. Integer addresses and dyadic arguments are computed exactly. The same argument applies to the row flag, using certified approximations of $\lambda_i/\overline{A}$ with its fixed slack.

We choose a working-bit bound polynomial in $d$, $\log N$, $\log L$, and $\log(1/\zeta)$. By retaining intermediate values, copying the output, and reversing the computation, Boolean arithmetic is made reversible in polynomial time and space. Bit-controlled rotations synthesize the angles with polynomial overhead in this bound and $\log(1/\zeta)$, using fixed-phase Clifford+$T$ approximation. We rely only on the existence of this overhead, not typical synthesis counts.

At every inverse call, we use the actual synthesized circuit in reverse. We compare the resulting full unitary sequence with its ideal sequence by telescoping, allocating the error per occurrence, including the $O(d+1)$ calls to the amplification block. This bounds the accumulated approximation error on all output registers with logarithmic extra precision in the occurrence count. The bounded extensions ensure uniform approximation on every address encountered during amplification.

Metadata remains loaded throughout the inner amplification sequence; loading, unloading, and the final interval scan each cost $O(N)$ times a polynomial in bit length. All remaining operations have polynomial cost in that length and $d$, proving the total gate and width bounds for the circuit and its inverse.

\section{Proof of Compatible Parameter Choice}\label{S-sec:param_choice}
\begingroup\begin{proof}
We prove main-text Lemma~\ref{lem:param-choice}.  For $N\geq2$ and
$0<\varepsilon\leq1/2$, certify $\overline C_W$ as in article
Equation~\eqref{eq:aw-certified-weight} and fix
$B_0\geq\sqrt{\overline C_W C_M}$ from the analytic bound.  Choose an
integer $d\geq1$ and then $\overline A$ as in article
Equation~\eqref{eq:aw-Abar}, with
\begin{equation}\label{S-eq:aw-precision-choices}
B_0(d+1)^2\sqrt N\,2^{-d}\leq\varepsilon/48,
 \qquad
 \tau=\frac{\varepsilon}{24\overline A}.
\end{equation}
The exponential decay permits $d=O(\log(N/\varepsilon))$, so
$\overline A=O(\sqrt{d+1})$ and $0<\tau\leq1/10$.
These choices precede the grid size: neither error coefficient in article
Equation~\eqref{eq:aw-error-coefficients} depends on $L$.

Choose a power-of-two $L\geq C(N+1)^3\tau^{-8}$, as required by
Theorem~\ref{S-thm:uniform-grid}, large enough that
\begin{equation}\label{S-eq:aw-grid-closure}
h\leq\min\left\{
  \frac{c}{16\sqrt{3N}},
  \frac{c}{8C_0C_1\sqrt{3N}},
  \frac{c\varepsilon}{48\sqrt3 B_0(1+C_0+C_1)N}
 \right\},
 \qquad \frac{Lh}{2}\geq\sqrt{2N}+2.
\end{equation}
The first bound permits the certified disjoint windows of article
Equation~\eqref{eq:aw-window-choice}.  The factor-four admissible
radius interval permits a dyadic choice with a margin; a root enclosure
narrower than $h/4$ determines a half-grid center within $h$ of the root.
The second bound gives $m_i\geq4C_0C_1$, since
$m_i>r_i/(2h)$ and $r_i\geq c/\sqrt{3N}$.  The last condition contains
every window in the grid.  The third bound and
\eqref{S-eq:aw-precision-choices} give $E_d(h)\leq\varepsilon/24$.

Every lower bound on $L=2\pi/h^2$ is polynomial in $N,1/\varepsilon$;
rounding the largest to the next power of two gives
$L=\operatorname{poly}(N,1/\varepsilon)$ and
$\log L=O(\log(N/\varepsilon))$.  Reselecting the windows on this
grid preserves $\overline A\geq2\|R\|$, because the bound in article
Equation~\eqref{eq:aw-Abar} holds for every admissible grid.
Finally $\overline A\tau=\varepsilon/24$, which is the separate sampler
error allocation used in the main proof.
\end{proof}
\endgroup

\section{Proof of Finite Precision Compilation}\label{S-sec:precision_compilation}
\begingroup\begin{proof}
We prove main-text Lemma~\ref{lem:finite-compilation} for the specific circuit
assembled there. Put
\[
 t=\lceil\overline A\rceil,\qquad M=2t+1,\qquad
 \zeta=\frac{\varepsilon}{4M}.
\]
The fixed-phase amplification sequence \eqref{S-eq:aw-outer-sequence}
has exactly $M$ forward or inverse block calls.  Each block comprises
the actual sampler $V_L$, the averaging unitary, and a coin whose
positive good amplitude is
$\overline A\sin(\pi/(4t+2))<\pi/4$.
Approximate the averaging unitary and the coin each to operator error
$\zeta$.  Article Lemma~\ref{lem:aw-row-state}, with the bounded
extensions and guard-bit analysis of Section~\ref{S-app:finite_precision},
includes the inner amplification, arithmetic, metadata scans, and
cleanup in the averaging-unitary budget.  The fixed coin slack makes
its angle uniformly conditioned.  The input and good-output reflections
test specified zero registers and fixed register locations; they have
exact Clifford+$T$ implementations with reusable clean workspace.

Use the same compiled $V_L$ in the comparison and implemented sequences.
Its full-output accuracy $\tau$ has already entered article
Equation~\eqref{eq:aw-good-error}; no exact implementation of the
nonisometric sampling matrix $J_L$ is assumed.  Every inverse call is
the actual reverse of the corresponding circuit, including $V_L$.
Thus each block differs from its comparison block by at most
$2\zeta$, and full-unitary telescoping adds at most
$2M\zeta=\varepsilon/2$ on all output registers.  Variations in row
success amplitudes and nonzero residual workspace are included in this
norm estimate, rather than discarded by a final projection.

Here $M=O(\sqrt{d+1})$ and
$\log(1/\zeta)=O(\log(N/\varepsilon))$. Each call uses scans of
$N$ records and arithmetic of polynomial logarithmic cost;
the sampler, coin, and reflections have polynomial logarithmic cost
as well.  The online logical Clifford+$T$ gate count is therefore
$O(N\operatorname{polylog}(N,1/\varepsilon))$, which also bounds
depth.  Serial scans and inverse computations reuse their workspace,
giving total width $\operatorname{polylog}(N,1/\varepsilon)$.

The $O(N)$ loaded records have polynomial logarithmic length.  Their
classical certification uses real-root isolation for the degree-$N$
integer polynomial $H_N$, interval evaluation of the weight formula,
and direct rational evaluation of the normalization constants $D_i$.
The coefficients of $H_N$ have $O(N\log N)$ bits; the root-gap
bound in main-text Lemma~\ref{lem:local-weight} and the dyadic margins
above make the required isolations finite with polynomial bit cost.
For the normalizations, write $K_d=P/Q$ with integer $P$ of degree at
most $d$; a common denominator for $D_i$ divides $Q^2m_i^{2d+1}$.
The Legendre recurrence and $D_i\leq d+5/4$ bound the denominator and
numerator bit lengths polynomially in $d+\log m_i$.
Since $m_i\leq L=\operatorname{poly}(N,1/\varepsilon)$, direct summation
therefore has polynomial classical bit cost.
Gaussian factors can require polynomially many classical guard bits
during certification; only the shorter certified records enter the
online circuit.  Classical certification and circuit-description
generation are polynomial in $N,1/\varepsilon$ and are charged
separately from the online gate and width bounds.
\end{proof}
\endgroup

\section{Coherent sampling of Hermite expansions on a uniform grid}
\label{S-qht:appendix}

Jain et al. construct a uniform-grid Hermite sampler for arbitrary
coefficient superpositions, including coherent degree-label erasure
\cite[Theorem~19, Lemma~27 and Algorithm~2]{jain2025}.  Their factor
$(-1)^k$ is removed by an input $Z$ on the least significant degree bit
when $N\geq2$, and by the identity when $N=1$.  We retain a direct
sampled-function proof of the clean interface used in the adaptive
conversion.  It follows their preparation, oscillator-phase selection,
and amplification strategy, without requiring the complete spectrum of a
discretized Hamiltonian.

Throughout, $L\geq8$ is dyadic,
$m=\log_2L$, $h=\sqrt{2\pi/L}$, and $I_L=[-L/2,L/2)$.
Grid labels $\ell\in I_L\cap\mathbb Z$ use a fixed reversible signed
encoding, such as $\ell\bmod L$.  For $1\leq d<L$, set
\begin{equation}\label{S-qht:sampling-matrix}
 J_L^{(d)}[\ell,k]=\sqrt h\,\varphi_k(h\ell),\qquad
 0\leq k<d,\qquad J_L=J_L^{(N)},
\end{equation}
where $\varphi_k(x)=e^{-x^2/2}H_k(x)/(2^k k!\sqrt\pi)^{1/2}$ has
positive leading coefficient.  Norms of finite-dimensional maps are
operator norms; absolute constants may increase between estimates.

\begin{theorem}[Clean uniform-grid sampling]\label{S-thm:uniform-grid}
There is an absolute constant $C$ with the following property.  Let
$N=2^n\geq1$ and $0<\tau\leq1/10$.  For every power of two
\begin{equation}\label{S-qht:grid-floor}
 L\geq C(N+1)^3\tau^{-8},
\end{equation}
there is a logical quantum circuit $V_L$ with clean input and output
embeddings $E_{\mathrm{in}}:\mathbb C^N\longrightarrow\hilbert_{\rm all}$
and $E_{\mathrm{grid}}:\mathbb C^L\longrightarrow\hilbert_{\rm all}$ such
that
\begin{equation}\label{S-qht:clean-interface}
 \bigl\|V_LE_{\mathrm{in}}-E_{\mathrm{grid}}J_L\bigr\|\leq\tau.
\end{equation}
The input embedding places the degree label in a zero-padded register and
initializes all other registers to zero.  The output embedding retains only
the grid register and sets the degree, phase, and arithmetic registers to
zero.  The logical gate count and total number of qubits are polynomial in
$\log L+\log(1/\tau)$.  The estimate includes the phases of all columns and
holds after tensoring with an arbitrary reference system.

The constant and the circuit family are qualitative asymptotic ones.  The
construction uses the classical bulk Hermite asymptotic and the standard
phase-preserving singular-vector amplification theorem stated below; it
does not specify a numerical value of $C$, an optimized gate exponent, or a
physical-hardware compilation.
\end{theorem}

\subsection{Sampled oscillator phases}

Let $\mathcal Ff(\xi)=(2\pi)^{-1/2}\int e^{-ix\xi}f(x)\,dx$ and
$F_L[j,\ell]=L^{-1/2}e^{-2\pi i j\ell/L}$.  Both Fourier transforms
have negative sign.  Define
\[
 C_L(a)=\diag_{\ell\in I_L\cap\mathbb Z}(e^{-ia(h\ell)^2}),\qquad
 A(u)=\tfrac12\tan(u/2),\qquad B(u)=\tfrac12\sin u,
\]
and the bounded half-time word
\begin{equation}\label{S-qht:six-chirp-word}
 \mathcal V_L(t)=
 \left[F_L^*C_L(A(t/2))F_L C_L(B(t/2))
 F_L^*C_L(A(t/2))F_L\right]^2,\qquad |t|\leq\pi.
\end{equation}
The following estimate controls its full output, including leakage from
the sampled Hermite space.

\begin{lemma}\label{S-qht:sampled-phases}
Let $J=J_L^{(d)}$ and $D_d(t)=\diag_{k<d}(e^{-it(k+1/2)})$.
There are absolute constants $C_1,C_2,c>0$ such that, with
$E(d,L)=C_1L^{1/4}e^{C_2d-cL}$,
\begin{equation}\label{S-qht:phase-frame-bound}
 \|J^*J-\id_d\|+\sup_{|t|\leq\pi}
 \|\mathcal V_L(t)J-JD_d(t)\|\leq E(d,L).
\end{equation}
For efficiently computable $t$, clean-workspace implementation of the word
to full-output error $\zeta$ uses
$\operatorname{poly}(\log L+\log(1/\zeta))$ gates and qubits.
\end{lemma}

\begin{proof}
Complex Mehler, Equation~\eqref{S-eq:aw-complex-mehler} in
Section~\ref{S-app:mehler} (see also \cite[Eq.~18.18.28]{dlmf}), gives,
with $a=\tanh(1/2)$, for $f=\sum_{k<d}\alpha_k\varphi_k$ and
$\|\alpha\|_2=1$,
\begin{equation}\label{S-qht:coefficient-envelope}
 |f(x+iy)|\leq Ce^{d/2-ax^2/2+y^2/(2a)}.
\end{equation}
For real $v$, the Hermite generating function and a Gaussian Fourier integral give
\[
 \mathcal F(e^{ivx^2/2}\varphi_k)(\xi)
 =s^{-1/2}e^{i\arctan(v)/2}\left(\frac{v-i}{s}\right)^k
 e^{-iv\xi^2/(2s^2)}\varphi_k(\xi/s),\qquad s=\sqrt{1+v^2}.
\]
The branch is continuous at $v=0$, fixing the scalar phase.  Thus Fourier
transformation preserves the family
$\sigma^{-1/2}e^{i\kappa x^2/2}e^{i(k\theta+\phi)}\varphi_k(x/\sigma)$,
where $\sigma>0$ and $\kappa,\theta,\phi\in\mathbb R$, with
\[
 \sigma'=\frac{\sqrt{1+\kappa^2\sigma^4}}{\sigma},\qquad
 \kappa'=-\frac{\kappa\sigma^4}{1+\kappa^2\sigma^4}.
\]
Inverse Fourier transformation adds a reflection, and a chirp adds its
coefficient to $\kappa$.  The finitely many operations in
\eqref{S-qht:six-chirp-word} have chirp coefficients of magnitude at most
one.  Their continuous parameter maps keep the widths bounded above and
away from zero and the chirps bounded.  Consequently every ideal continuum prefix $g$ and its
next Fourier or inverse Fourier transform satisfy
$|g(x)|+|\mathcal Fg(x)|+|\mathcal F^{-1}g(x)|\leq Ce^{d/2-cx^2}$,
uniformly in $t$ and $\alpha$.

Write $S_Lg=\sqrt h(g(h\ell))_{\ell\in I_L\cap\mathbb Z}$.
Poisson summation, with $2\pi/h=Lh$, gives
\[
 (F_LS_Lg-S_L\mathcal Fg)_j
 =\sqrt h\sum_{q\ne0}\mathcal Fg(hj+qLh)
 -\sqrt{h/L}\sum_{\ell\notin I_L}g(h\ell)e^{-2\pi i j\ell/L}.
\]
The Gaussian bounds and
$|hj+qLh|\geq(2|q|-1)Lh/2$ show that its norm is at most
$CL^{1/4}e^{d/2-cL}$; the inverse transform obeys the same estimate.
Sampling commutes exactly with a chirp.  Telescope the eight Fourier
steps using these bounds on ideal prefixes and unitary grid suffixes.
The corresponding continuum word equals $e^{-itH_{\rm osc}}$, where
$H_{\rm osc}=(\widehat p^2+\widehat x^2)/2$,
$\widehat p=-i\,d/dx$, and $\widehat x f=xf$: apply the three-chirp
identity \cite[Theorem~3, Eq.~(6)]{jain2025} twice at $t/2$, choosing
the continuous branch equal to the identity at time zero.  Since
$H_{\rm osc}\varphi_k=(k+1/2)\varphi_k$, this proves the phase bound.

For the Gram bound, the entire function
$G_f(z)=f(z)\overline{f(\overline z)}$ satisfies
$|G_f(x+iy)|\leq Ce^{d-ax^2+y^2/a}$.  Shifting its Fourier contour to
$y=-a\omega/2$ gives
$|\mathcal FG_f(\omega)|\leq Ce^{d-a\omega^2/4}$; Gaussian decay
removes the vertical boundary terms.  Poisson summation and removal of
the spatial tails therefore give
$h\sum_{\ell\in I_L}|f(h\ell)|^2=1+O(e^{d-cL})$.
The uniform Rayleigh-quotient bound proves the Gram estimate.

Finally, a Fourier circuit has $O(m^2)$ gates before finite-set synthesis.
A chirp computes the signed square, applies its phase, and uncomputes.
Since $(h\ell)^2=O(L)$, coefficient accuracy $O(\zeta/L)$ suffices.
Arithmetic, bounded trigonometric evaluation and rotation synthesis
\cite{ross2016} have polynomial cost in
$\log L+\log(1/\zeta)$; allocate error across their gates.
Compute--phase--uncompute gives the stated clean-input, full-output bound.
\end{proof}

For analysis only, when $\gamma=\|J^*J-\id_d\|\leq1/2$, set
\[
 G=J^*J,\qquad Q=JG^{-1/2},\qquad q_k=Q\ket{k}.
\]
Functional calculus on $[1/2,3/2]$ gives $Q^*Q=\id_d$ and
$\|Q-J\|\leq C\gamma$.  Adding and subtracting $J$ in
\eqref{S-qht:phase-frame-bound} shows that a compiled word
$\widetilde{\mathcal V}_L(t)$ has error
$\delta=C(E(d,L)+\gamma+\zeta)$ relative to $QD_d(t)$ on this frame,
including all output workspace.  No commutation of $G$ and $D_d(t)$ is
needed, and $Q$ is not computed by the circuit.

\subsection{Preparation and clean degree-label erasure}

\begin{proof}[Proof of Theorem~\ref{S-thm:uniform-grid}]
We construct smooth seeds, select their Hermite component, and amplify
before erasing the retained degree label.

\emph{Seeds and sampled tails.}
Szeg\H{o}'s bulk formula \cite[Theorem~8.22.9(a), Eq.~(8.22.12),
p.~201]{szego1975}, in our positive-leading normalization, is
\begin{equation}\label{S-qht:pr-input}
 \varphi_k(\sqrt{2k+1}\cos\theta)
 =\frac{2^{1/4}\sin((k/2+1/4)(\sin2\theta-2\theta)+3\pi/4)}
 {\sqrt\pi\,k^{1/4}\sqrt{\sin\theta}}+O(k^{-5/4}),
\end{equation}
uniformly on fixed compact subintervals of $(0,\pi)$.
Put $a_k=\sqrt{3(2k+1)/4}$ and $w_k=(10\sqrt{2k+1})^{-1}$.
Let $g_k$ equal one on $|x|\leq a_k$, zero on $|x|\geq a_k+w_k$,
and $1-10s^3+15s^4-6s^5$ in between, with $s=(|x|-a_k)/w_k$.
Its support stays in a fixed angular bulk interval, since
$|x|/\sqrt{2k+1}\leq\sqrt3/2+1/10<1$.
For $k\geq k_0$, define $f_k$ as $g_k$ times the explicit sine expression
in \eqref{S-qht:pr-input}; for the finitely many $k<k_0$, use
$f_k=g_k\varphi_k$.  Set $f_k=0$ outside the support and evaluate the
bulk expression only inside it.

The cutoff is $C^2$.  Direct differentiation of the explicit expressions,
not of the asymptotic remainder, gives
\begin{equation}\label{S-qht:seed-bounds}
 \|f_k^{(r)}\|_\infty\leq C(k+1)^{r/2-1/4}\ (0\leq r\leq2),
 \qquad \tfrac12\int(|f_k'|^2+x^2|f_k|^2)\,dx\leq C(k+1).
\end{equation}
Here the support length is $O(\sqrt{k+1})$ and both the phase derivative
and inverse cutoff width are $O(\sqrt{k+1})$.
The moment $\int x^2\varphi_k^2=k+1/2$ implies
$\int_{-a_k}^{a_k}\varphi_k^2\geq1/3$.
The bulk remainder gives $\|g_k(\widetilde f_k-\varphi_k)\|_2\leq C/k$,
where $\widetilde f_k$ is the sine expression.  Choosing $k_0$ large
enough yields $\langle\varphi_k,f_k\rangle\geq b>0$ and
$b\leq\|f_k\|_2\leq C$ for all $k$, including zero.
For a sufficiently large absolute constant $C_s$, set
$B_k=(k+1)^{1/4}f_k/C_s\in[-1/2,1/2]$.
Elementary-function evaluation with angle precision
$O(\sigma/(k+1))$ and guard bits evaluates $B_k$ to error $\sigma$ with reversible
resources polynomial in $\log L+\log(k+1)+\log(1/\sigma)$.
The fixed-degree exceptional branch and the $C^2$ cutoff joins have the
same bound; no growing table is required.

For $N\leq d<L$ and $0\leq k<N$, put
$F_k=\sqrt h\sum_{\ell\in I_L}f_k(h\ell)\ket\ell$,
sampling the entire seed support.
Write $c_j=\langle\varphi_j,f_k\rangle$ and $\alpha=J^*F_k$.
The oscillator quadratic form gives
$\sum_{j\geq d}|c_j|^2\leq C(k+1)/d$.
The ladder relations and
$\|u\|_\infty^2\leq2\|u\|_2\|u'\|_2$ give
$\|\varphi_j^{(r)}\|_\infty\leq C_r(j+1)^{r/2+1/4}$.
Leibniz's rule and full-support trapezoidal quadrature thus yield
\[
 |\alpha_j-c_j|\leq e_{\rm q}:=\frac{C\sqrt{N+1}(d+1)^{5/4}}L,
 \qquad \big|\|F_k\|^2-\|f_k\|_2^2\big|\leq\frac{C(N+1)}L.
\]
There is no endpoint term.  The exact identity
$\|(\id-QQ^*)F_k\|^2=\|F_k\|^2-\alpha^*G^{-1}\alpha$,
$\|G^{-1}-\id\|\leq C\gamma$, and the squared-coefficient error
$O(d e_{\rm q})$ give
\begin{equation}\label{S-qht:sampled-tail}
 \left\|(\id-QQ^*)\frac{F_k}{\|F_k\|}\right\|^2
 \leq C\left(\frac{N+1}{d}+\frac{(d+1)^{11/4}}L+\gamma\right).
\end{equation}
For the small errors chosen below, $c\leq\|F_k\|\leq C$ and
$b_k:=\langle q_k,F_k/\|F_k\|\rangle\geq c>0$, because
$\langle q_k,F_k\rangle=c_k+O(e_{\rm q}+\gamma)$.
All these values and matrices are real, fixing the overlap phase.

\emph{Selection and signed preparation.}
Let $U$ be phase estimation using all $m$ controlled providers
$e^{it_j/2}\widetilde{\mathcal V}_L(t_j)$,
$t_j=2\pi2^j/L$, $j=0,\ldots,m-1$, and the positive-sign Fourier
transform on its phase register.  The scalar is a relative phase on the
control-one branch.  Extend the ideal shifted action
$q_k\mapsto e^{-it_jk}q_k$ by the identity on
$(\operatorname{ran}Q)^\perp$, obtaining an ideal circuit $U_0$.
Character orthogonality gives $U_0(\ket0q_k)=\ket{k}q_k$ since $d<L$.
Its intermediate operations preserve the space $\mathcal M$ with
arbitrary phase register, grid in $\operatorname{ran}Q$, and zero work.
Unitary hybrids, including Fourier synthesis, give
\begin{equation}\label{S-qht:qpe-full-bounds}
 \|(U-U_0)|_{\mathcal M}\|\leq Cm\delta,
 \qquad \|(U^*-U_0^*)|_{\mathcal M}\|\leq Cm\delta.
\end{equation}
For the inverse use $U^*-U_0^*=U^*(U_0-U)U_0^*$ and
$U_0^*\mathcal M=\mathcal M$.  Actual frame invariance is unnecessary.

Retaining $k$, prepare a uniform dyadic interval of $T_k$ grid points
containing the whole seed support plus one mesh step, with
$T_kh=\Theta(\sqrt{k+1})$.  Conservative binary endpoints suffice.
Compute $2\arcsin B_k(h\ell)$, apply the corresponding $R_y$ rotation
to a flag initially in $\ket0$, then apply $X$ and uncompute evaluation.
With $R_y(2\theta)\ket0=\cos\theta\ket0+\sin\theta\ket1$, the
zero-flag amplitude is exactly $B_k(h\ell)$, including its sign.
The fixed slack $|B_k|\leq1/2$ bounds the derivative of $\arcsin$,
so angle approximation and rotation synthesis give full input-isometry
error $O(\sigma)$ with the same polylogarithmic resources.
The selected vector is
$\kappa_kF_k+O(\sigma)$, where
\[
 \kappa_k=\frac{(k+1)^{1/4}}{C_s\sqrt{T_kh}},\qquad
 0<c\leq\kappa_k\leq C.
\]
The error is independent of $T_k$ because the initial interval is
normalized.  Append $U$, preserving the raw failure flags, and subtract
$k$ from the phase register.  Call this compiled unitary $A$ and let
$\Pi_{\rm good}$ require the zero raw flag and zero phase and work
registers.  Suppressing the retained label, put
$z_k=\Pi_{\rm good}A_k\ket0$.
If the tail norm in \eqref{S-qht:sampled-tail} is at most $\eta$, then
\begin{equation}\label{S-qht:good-vector}
 \|z_k-\beta_k(q_k\otimes\ket0)\|\leq
 C(\sigma+\eta+m\delta)=:\eta_g,\qquad
 \beta_k=\kappa_k\|F_k\|b_k\geq a_*>0.
\end{equation}
Indeed, ideal selection extracts exactly the $q_k$ coefficient of the
frame component, and unitarity bounds the tail contribution by its norm.
Raw failure branches cannot enter the good subspace.

\emph{Amplification and erasure.}
We use phase-preserving singular-vector amplification
\cite[Theorems~26--27 and Corollary~10, arXiv v1]{gilyen2018}:
if the nonzero singular values of $\Pi_{\rm good}A\Pi_{\rm in}$ lie
in $[s_*,1]$, a common polynomial gives real overlap at least $1-\rho$
with each normalized good singular vector, using
$O(s_*^{-1}\log(1/\rho))$ calls to $A,A^*$ and projector phases.
This external input fixes phases, not just success probabilities.
Apply it to the actual $A$ and its exact circuit inverse, with
$\Pi_{\rm in}$ selecting $k<N$ and zero input auxiliaries.
For $\eta_g\leq a_*/2$, the actual success amplitude
$s_k=\|z_k\|\geq a_*/2$ and direction $v_k=z_k/s_k$ satisfy
$\|v_k-q_k\otimes\ket0\|\leq4\eta_g/a_*$.
Retained labels make these selected columns orthogonal, with singular
values exactly $s_k$.  Choose $\rho=\varepsilon_{\rm amp}^2/8$:
real overlap $1-\rho$ gives vector error at most $\sqrt{2\rho}$.
Allocate another $\varepsilon_{\rm amp}/2$ to amplification-phase
synthesis.  Preparation of $\ket{k}q_k\ket0$ then has operator error
at most $\varepsilon_{\rm amp}+4\eta_g/a_*$, taking the maximum over
orthogonal label blocks.  Seed and selection errors enter once through
the actual good direction, not once per amplification call.

Apply the exact compiled $U^*$ using the retained degree register as its
phase register.  For arbitrary coefficients $\lambda_k$, its ideal
inverse sends $\sum_{k<N}\lambda_k\ket{k}q_k$ to
$\ket0\sum_{k<N}\lambda_kq_k$.
Equation~\eqref{S-qht:qpe-full-bounds} costs $Cm\delta$; unitarity
preserves the preceding error.  Replacing $Q$ by $J$ costs $C\gamma$.
Thus the full-output error is at most
\begin{equation}\label{S-qht:total-error}
 \varepsilon_{\rm amp}+C(\sigma+\eta+m\delta)/a_*
 +Cm\delta+C\gamma.
\end{equation}

\emph{Parameters and resources.}
Choose $\sigma,\eta,\varepsilon_{\rm amp}$ as sufficiently small fixed
multiples of $\tau$, then take
\[
 d=\lceil C_3(N+1)\tau^{-2}\rceil,\qquad
 L\geq C_4(d+1)^3\tau^{-2},\qquad
 \zeta\leq c_0\tau/(1+m).
\]
Increasing $C_3,C_4$ and decreasing $c_0$ makes
\eqref{S-qht:sampled-tail} at most $\eta^2$ and
\eqref{S-qht:total-error} at most $\tau$.
Indeed, the grid bound ensures full support containment, $d<L$, and
quadrature error $O(\tau^2)$, and also
$E(d,L),\gamma\leq c_0\min\{\tau^2,\tau/(1+m)\}$.
For this last assertion keep $d$ fixed as $L$ increases: half of $-cL$
absorbs $C_2d$, and the other half dominates $L^{1/4}(1+\log L)$ and
the required powers of $1/\tau$ already at the polynomial floor.
These bounds decrease beyond a fixed threshold.  Hence
\eqref{S-qht:grid-floor}, with a sufficiently large absolute constant,
works for every larger dyadic grid.

Each selection uses $m$ providers, and amplification uses
$O(\log(1/\tau))$ calls.  Seed evaluation, label arithmetic,
projector tests, finite-gate synthesis with occurrence-wise allocation,
and exact circuit inverses have polynomial cost in
$\log L+\log(1/\tau)$.  Evaluation is uncomputed; all temporary,
phase and amplification registers have the same polynomial width.
There is no unit-cost state-table oracle.  The estimates include every
auxiliary output and are operator bounds, so they survive tensoring with
an arbitrary reference system.  This proves the clean interface and its
logical resource claims.

The fixed thresholds and classical-input constants remain qualitative.
Circuit-description generation is separate from the online quantum count;
no numerical constant, optimized classical compiler or hardware
realization is asserted.  A polynomially chosen grid gives polynomial
dependence on $\log N+\log(1/\tau)$, while larger grids retain the
explicit $\log L$ dependence.
\end{proof}

\section{Sampling outside the finite Hermite space}\label{S-sec:input-error}
An exact finite unitary need not be an exact continuous transform. Let $P_N$ be the orthogonal projection of $L^2(\mathbb R)$ onto $\hilbert_N$. For a function $\psi$ with well-defined samples, write $r=\psi-P_N\psi$ and let $c_{<N}$ be the coefficient vector of $P_N\psi$. The finite output is
\begin{equation}\label{S-eq:aliasing-error}
 BS_N\psi=c_{<N}+BS_Nr,\qquad
 \|BS_N\psi-c_{<N}\|_2=\|S_Nr\|_2.
\end{equation}
Moreover, interpreting the output as a function in $\hilbert_N$ gives the orthogonal error decomposition
\begin{equation}\label{S-eq:total-function-error}
 \left\|\sum_{k=0}^{N-1}(BS_N\psi)_k\varphi_k-\psi\right\|_{L^2}^2
 =\|S_Nr\|_2^2+\|r\|_{L^2}^2.
\end{equation}
The first term is aliasing from the sampled tail; the second is truncation. These identities follow directly from \eqref{eq:unitary_B} and orthogonal projection. They do not assume that point evaluation is bounded on $L^2$.

Two immediate examples show why a sampling assumption is necessary. Since every $x_i$ is a root of $H_N$, $S_N\varphi_N=0$. Applying the Hermite recurrence \cite[Sec.~18.9]{dlmf} at those roots also gives
\begin{equation}\label{S-eq:first-alias}
 BS_N\varphi_{N+1}=-\sqrt{\frac{N}{N+1}}\,e_{N-1},
\end{equation}
where $e_{N-1}$ is the final number-basis coordinate. Thus a component just outside the cutoff can disappear entirely or alias into a retained component.

There is no constant $C_N$ that bounds $\|S_Nr\|_2$ by $C_N\|r\|_{L^2}$ for every smooth $r\perp\hilbert_N$. To see this, place a smooth bump of width $h$ and height $h^{-1/2}$ at one quadrature node, with fixed nonzero central value after rescaling. Its $L^2$ norm is fixed and its sampled value diverges. Subtract its projection onto the fixed finite space. Each projection coefficient is $O(h^{1/2})$, so the subtraction neither cancels the diverging sample nor makes the residual norm diverge. This is the usual unboundedness of point evaluation on $L^2$, specialized to the finite sampling operator.

A stronger input assumption does provide a direct bound. If $r\in H^1(\mathbb R)$, then its continuous representative satisfies $\|r\|_\infty^2\leq\|r\|_{L^2}\|r'\|_{L^2}$. One proof integrates $(|r|^2)'$ from both sides of a point and applies Cauchy--Schwarz. Consequently, with $W_N=\sum_iw_i e^{x_i^2}$,
\begin{equation}\label{S-eq:sobolev-sampling}
 \|S_Nr\|_2^2\leq W_N\|r\|_{L^2}\|r'\|_{L^2}.
\end{equation}
For the oscillator $H_{\mathrm{osc}}=(-d^2/dx^2+x^2)/2$, its quadratic form is
\[
 q_{\mathrm{osc}}[r]=\tfrac12\bigl(\|r'\|_{L^2}^2+\|xr\|_{L^2}^2\bigr),
 \qquad r\in H^1(\mathbb R),\quad xr\in L^2(\mathbb R).
\]
Thus $\|r'\|_{L^2}^2\leq2q_{\mathrm{osc}}[r]$ on the form domain. On the smaller operator domain, $q_{\mathrm{osc}}[r]=\langle r,H_{\mathrm{osc}}r\rangle$ as well; the sampling bound requires only the form-domain assumption.
Equations \eqref{S-eq:total-function-error} and \eqref{S-eq:sobolev-sampling} give a sufficient error bound. Determining sharper $N$-dependence under an appropriate Hermite regularity assumption is a further approximation-theory question.

For an actual quantum input one must normalize $S_N\psi$. If $a=\|c_{<N}\|_2>0$ and $e_{\rm tail}=\|S_Nr\|_2<a$, then the normalized finite output and the normalized projected coefficient state satisfy
\begin{equation}\label{S-eq:normalized-error}
 \left\|\frac{BS_N\psi}{\|S_N\psi\|_2}-\frac{c_{<N}}a\right\|_2\leq\frac{2e_{\rm tail}}a.
\end{equation}
This follows by adding and subtracting $BS_N\psi/a$ and using the reverse triangle inequality. A circuit approximation of operator-norm error $\varepsilon$ adds at most $\varepsilon$ to this state-vector error. If $S_N\psi=0$, no normalized sampled input exists.

For the oscillator application, evolution and measurement already specified in the number basis can be performed directly in that basis. The scalar $e^{-it/2}$ in the position-basis evolution becomes a relative control phase in a controlled implementation. More general multiplication potentials require a separate approximation: replacing a potential by its values at the quadrature nodes is a discrete-variable-representation approximation, whose error must be controlled independently of the exact finite transform~\cite[Appendix~A, Eqs.~(42)--(45)]{plis2025}.
A direct composition check can form $d=B^\dagger c$ for a normalized coefficient vector and compare $B^\dagger D(t)Bd$ with $B^\dagger D(t)c$. This tests the conventions and composition; compiled gate counts require the separate resource analysis.

\section{A Mode-Population Noise Diagnostic}
\label{S-app:noise-diagnostic}
\label{S-sec:noise_benchmarking}
A transform--idle--inverse experiment can show how a qubit noise channel redistributes Hermite-mode populations. The following protocol describes a possible diagnostic. The role of $B$ is to choose the preparation and measurement bases, rather than to establish a noise-characterization advantage. In an ideal-transform model, noise is confined to the idle interval by assumption; a hardware implementation would also require characterization of state preparation, transform errors, and readout.

For example, consider a 4-qubit register ($N=16$), initially in the number state $\ket{12}$. With the convention~\eqref{eq:B_eq}, $B^\dagger$ maps this state to weighted position samples; after an idle interval $\Delta t$, $B$ returns it to the number basis for measurement. If $\Phi_{\Delta t}$ is the noise channel on the sample register, the measured state is
{\begin{equation}\label{S-eq:noise-mode-channel}
 \rho_{\rm out}
 =B\,\Phi_{\Delta t}\bigl(B^\dagger\rho_{\rm in}B\bigr)B^\dagger,
 \qquad \rho_{\rm in}=\ket{12}\bra{12}.
\end{equation}}
Figure~\ref{S-fig:benchmark_circuit} gives this protocol.

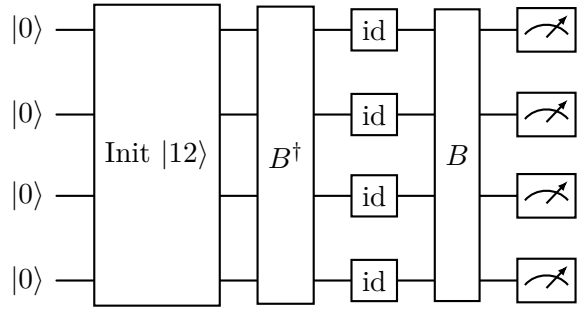
\begin{figure}[ht]
    \centering
    \begin{quantikz}
        \lstick{$\ket{0}$} & \gate[wires=4]{\text{Init } \ket{12}} & \gate[wires=4]{B^\dagger} & \gate{\text{id}} & \gate[wires=4]{B} & \meter{} \\
        \lstick{$\ket{0}$} & & & \gate{\text{id}} & & \meter{} \\
        \lstick{$\ket{0}$} & & & \gate{\text{id}} & & \meter{} \\
        \lstick{$\ket{0}$} & & & \gate{\text{id}} & & \meter{}
    \end{quantikz}
    \caption{Transform--idle--inverse protocol. The number state is mapped to weighted position samples by $B^\dagger$, subjected to noisy idle (id) operations for time $\Delta t$, and mapped back by $B$ for number-basis measurement.}
    \label{S-fig:benchmark_circuit}
\end{figure}

The observable populations are $p_k(\Delta t)=\bra{k}\rho_{\rm out}\ket{k}$, whose sum is one for a trace-preserving channel. To illustrate this protocol, take $\Phi_t=\phi_t^{\otimes4}$, where each qubit undergoes thermal relaxation with $T_1=50\,\mu\mathrm{s}$, coherence-decay time $T_2=30\,\mu\mathrm{s}$, and equilibrium excited-state population $p_{\rm eq}=0.05$. For a single-qubit density matrix $\sigma$, the channel is specified by
{\begin{equation}\label{S-eq:noise-single-qubit}
 \begin{aligned}
  (\phi_t(\sigma))_{11}
    &=p_{\rm eq}+(\sigma_{11}-p_{\rm eq})e^{-t/T_1},\\
  (\phi_t(\sigma))_{01}&=\sigma_{01}e^{-t/T_2},
 \end{aligned}
\end{equation}}
with the remaining entries fixed by trace one and Hermiticity. We use the binary ordering $\ket{q_3q_2q_1q_0}$, with $q_0$ least significant. Figure~\ref{S-fig:noise_benchmark} is a direct density-matrix calculation of~\eqref{S-eq:noise-mode-channel} at integer-microsecond intervals, using ideal dense transforms. Qubit relaxation in this encoding need not relax the oscillator to its ground state: even a complete reset of the sample register gives $B\ket{0}$ in the number basis, rather than $\ket{0}$. Here the limiting state is $B\diag(1-p_{\rm eq},p_{\rm eq})^{\otimes4}B^\dagger$.

The accompanying reproducibility package contains the simulator, the full population table, and numerical consistency checks.

\begin{figure}[ht]
    \centering
    \includegraphics[width=0.8\linewidth,alt={Heat map of sixteen Hermite-mode populations versus idle time. Noise redistributes the initial mode twelve population; each row sums to one.}]{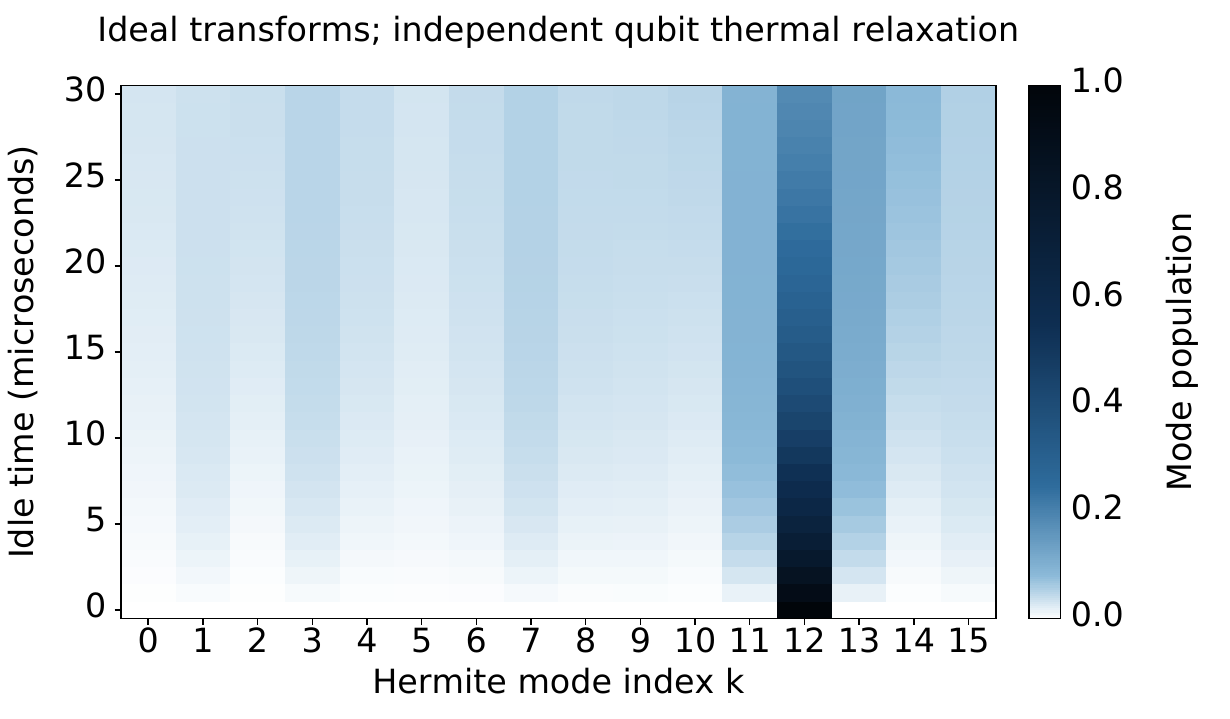}
    \caption{Simulated mode populations for the transform--idle--inverse protocol with $N=16$, initial mode $12$, and independent identical qubit channels~\eqref{S-eq:noise-single-qubit}. The parameters are $T_1=50\,\mu\mathrm{s}$, $T_2=30\,\mu\mathrm{s}$ and $p_{\rm eq}=0.05$. Each row gives all 16 populations at one idle time. The calculation uses ideal dense transforms and density matrices, without shot sampling, gate compilation, or hardware execution.}
    \label{S-fig:noise_benchmark}
\end{figure}

\subsection{Population and coherence tests}
Population measurements determine only part of a noise channel. The distinction is already visible for bosonic pure loss, whose single lowering operator gives the generator $\kappa\mathcal D[a_N]$ on the retained number states~\cite[Sec.~I.A, preceding Eq.~(1.3)]{albert2018}. Here $\kappa\geq0$ is the loss rate and
\[
 a_N=\sum_{k=1}^{N-1}\sqrt{k}\ket{k-1}\bra{k},\qquad
 \mathcal D[L](\rho)=L\rho L^\dagger-\tfrac12\{L^\dagger L,\rho\}.
\]
Its infinitesimal population transition $k\to k-1$ has rate $\kappa k$. Independent qubit relaxation instead has local lowering operators: in binary encoding, decay of bit $j$ changes the encoded integer by $2^j$, at a device-dependent rate. A basis change conjugates the channel as in~\eqref{S-eq:noise-mode-channel}; a model of encoded qubit noise and a model of bosonic loss therefore require distinct physical assumptions.

To see why populations alone are insufficient, put $E_{k\ell}=\ket{k}\bra{\ell}$ and compare
\[
 \mathcal L_1=\kappa\mathcal D[a_N],\qquad
 \mathcal L_2=\kappa\sum_{k=1}^{N-1}k\mathcal D[E_{k-1,k}].
\]
Both preserve number-diagonal density matrices and act identically on them, so every number-state input has the same population evolution under both generators at every time. For distinct $1\leq k,\ell<N$, however,
\begin{equation}\label{S-eq:coherence-witness}
 (\mathcal L_1-\mathcal L_2)(E_{k\ell})
 =\kappa\sqrt{k\ell}\,E_{k-1,\ell-1}.
\end{equation}
The anticommutator terms agree, while the collective jump transports off-diagonal coherence. For $N\geq3$ and $\kappa>0$, the state $(\ket1+\ket2)/\sqrt2$ distinguishes the two generators. They coincide for $N\leq2$ or $\kappa=0$.

Number-state inputs probe population leakage, while superposition inputs with phase-sensitive readout can probe coherence transport. A zero-idle round trip supplies a baseline for the combined preparation, transform, and readout errors. These are proposed extensions of the diagnostic; the population simulation above provides neither a coherence measurement nor a device calibration.

\clearpage

\end{document}